\documentclass[12pt]{article}
\usepackage[utf8]{inputenc}
\usepackage{amsfonts,amsmath}
\usepackage{natbib}
\usepackage{geometry}
\usepackage{graphicx}
\usepackage{setspace}
\usepackage{color}

\usepackage{amsthm}
\usepackage{amssymb}
\usepackage{hyperref}
\usepackage{multirow}
\hypersetup{colorlinks=true, linkcolor=blue, citecolor=blue} 

\usepackage{threeparttable}
\usepackage{bbm}

\newtheorem{definition}{Definition}
\newtheorem{lemma}{Lemma}
\newtheorem{thm}{Theorem}
\newtheorem{prop}{Proposition}
\newtheorem{remark}{Remark}
\newtheorem{example}{Example}
\def\argmax{\mathop{\rm arg\,max}}

\newtheorem{assumption}{Assumption}
\newtheorem{algorithm}{Algorithm}

\newcommand{\E}{\mathcal E} 
\newcommand{\cS}{\mathcal S}

\newcommand{\bne}[1]{Q_{\theta,#1}^{\mathit{BNE}}(x)}
\newcommand{\bce}[1]{Q_{\theta,#1}^{\mathit{BCE}}(x)}

\newcommand{\game}[1]{\Gamma^x(\theta,#1)}

\newcommand{\es}[2]{S^{#2}(#1)}

\newcommand{\ex}{\succeq_{Exp}}
\newcommand{\is}{\succeq}

\newcommand{\est}[1]

\newcommand{\wnorm}[1]{\|#1\|_{W_{x}}}
\newcommand{\ball}{\mathbb B_{x}}
\newcommand{\wnormn}[1]{\|#1\|_{W_{n,x}}}
\newcommand{\balln}{\mathbb B_{n,x}}

\newcommand{\vn}{V_{n,x}(\theta)}

\newcommand{\sv}{\Psi}

\usepackage{xcolor}

\title{Testing Information Ordering for Strategic Agents\thanks{We thank Camilla Roncoroni for inputs on an earlier version of this manuscript. We appreciate Gaurab Aryal, Debopam Bhattacharya, Tim Christensen, Phil Haile, Quang Vuong, J\"orgen Weibull and participants at the Cowles Conference, the Stochastic Dominance and Quantile-Based Methods in Financial Econometrics Workshop 2025 honoring Professors Oliver Linton and Yoon-Jae Whang, the Cemmap Workshop on Econometrics, Models of Strategic Interactions, and SWERIE 2025 for useful feedback. Kaido gratefully acknowledges financial support from NSF grants SES-1824344 and SES-2018498.}}

\author{Sukjin Han \\ School of Economics \\ University of Bristol \\ \url{vincent.han@bristol.ac.uk} \and 
Hiroaki Kaido \\ Department of Economics \\ Boston University \\ \url{hkaido@bu.edu} \and 
Lorenzo Magnolfi\\ Department of Economics \\ University of Wisconsin \\ \url{magnolfi@wisc.edu} }

\date{January 2026}

\begin{document}

\maketitle
\begin{abstract}
Specifying the information structure in strategic environments is difficult for empirical researchers. We develop a test of \emph{information ordering} that examines whether the true information structure is at least as informative as a proposed baseline. Using Bayes Correlated Equilibrium (BCE), we translate the ordering of information structures into testable moment inequalities and establish uniform asymptotic validity for our testing procedure. In an application to U.S.\ airline markets, we test whether hub airlines have informational advantages beyond cost and demand benefits. We reject the privileged information hypothesis, with rejections concentrated in large, competitive markets.\looseness=-1

\vspace{0.2in}

\end{abstract}

\clearpage
\onehalfspacing
\section{Introduction}

Many important economic interactions are strategic in nature. When analyzing data generated in such contexts, researchers often bring to data models of games to estimate primitives and perform counterfactual simulations. Examples include the analysis of firms' entry decisions \citep[e.g.,][]{berry1992estimation, ciliberto2009market} and bidding behavior in auctions \citep[e.g.,][]{haile2003inference}. One such key primitive is the information structure, i.e., a full description of the information available to players as they interact and generate the data. Information plays an important role in shaping outcomes: for example, potential entrants may receive a common signal about the unobserved profitability of a market, or certain bidders may know more than others about the realized valuations of bidders in an auction. Therefore, the information structure may be either of independent economic interest---e.g., because asymmetries in information may suggest that some firms have political connections \citep{magnolfi2016political,baltrunaite2020political} or engage in corruption---or crucial in evaluating counterfactual policy, as misspecifying the information structure may lead to misleading predictions.

The information structure that prevails in the strategic interaction is seldom known to the researcher. Thus, the researcher may need to infer information based on data on observable outcomes. However, the estimation of information structure is typically challenging. First, the space of possible information structures is potentially very large, as it encompasses any form of signals, with different degrees of informativeness, that players may receive. Restrictive parametrizations of information structures are thus likely to result in misspecification. Treating the information structure as a nonparametric object would be an alternative, but the high dimensionality of this object makes it infeasible with standard data used in applied research.

To make progress on this problem, we focus on testing hypotheses on a specific ordering of information structures. We consider models of discrete games, and adopt the orderings of information structures defined in \cite{Bergemann:2016aa}. Their results allow us to compare information structures and establish a precise notion of the \emph{ordering} of information structures.  Crucially, their results also imply that different levels of information have distinct implications for the observables, thus providing a basis for testing.  We can thus formulate statistical hypotheses to test whether the information structure that prevails in the data exceeds a certain baseline. 

Consider players who have access to an unknown information structure $\es{x}{}$ in a game with observable characteristics $x$, where the structure governs the precision of signals each player receives about underlying payoff states. Players play a Bayes-Nash equilibrium under $\es{x}{}$, and the researcher observes only their equilibrium actions. Rather than attempting to specify $\es{x}{}$ directly---a task that raises dimensionality and misspecification concerns---the researcher formulates a \emph{baseline} information structure $S^r(x)$ representing a minimal level of information that players are believed to possess, and tests whether the true structure is at least as informative. \cite{Bergemann:2016aa} show that, under a specific notion of ordering called \emph{individual sufficiency}, the ordering of information structures translates into an ordering of equilibrium predictions. Their solution concept, \emph{Bayes Correlated Equilibrium (BCE)}, is central to our approach for two reasons. First, BCE encompasses BNE in terms of predictions: if data are generated from a BNE under $\es{x}{}$, the resulting \emph{conditional choice probability} belongs to the set of BCE conditional choice probabilities under baseline $S^r(x)$. Second, the set of BCE predictions is computationally tractable because it is convex, characterized by linear inequalities and equalities.

We use the properties above to construct a test statistic for the null hypothesis that $S(x)$ that generates the conditional choice probability is as informative as $S^r(x)$. Our test statistic is based on the distance from the empirical conditional choice probability to the BCE prediction under the baseline information structure. This construction exploits the isometry between convex sets and their support functions. The null hypothesis is rejected when the test statistic exceeds a bootstrap critical value. The critical value is calculated so that the test is asymptotically valid uniformly over a large class of data-generating processes. According to our Monte Carlo experiments, the proposed test properly controls its size and can detect violations of the restrictions imposed by the null hypothesis.  Specifically, in an incomplete information game, we tested the null hypothesis that some payoff shocks are known to all players. The test showed monotonically increasing power when the conditional choice probabilities deviate locally from the BCE prediction under the null hypothesis. We also extend our approach to testing a sequence of hypotheses, which can help refine the understanding of a game's information structure. This approach applies multiple testing methods to the specification test with a partially identified model. By introducing the ordering of information structures in the sequence of hypotheses, we attain a simple powerful procedure for valid multiple testing without relying on conservative correction.

We apply our test to the U.S.\ airline industry to examine whether hub status confers informational advantages in addition to market power. Following deregulation in the 1970s, the industry rapidly evolved toward a hub-and-spoke network structure, where major airlines concentrate operations at specific airports. The economics literature has established that hub operations confer significant cost and demand advantages \citep{Berry:1996aa}. However, whether hub presence also provides superior information about market conditions and about the potential profitability of competitors remains an open question. Using data from the Department of Transportation's Origin and Destination Survey (DB1B) for the first quarter of 2019, we analyze 1,810 hub markets---routes where at least one endpoint serves as a hub for a major airline. We aggregate airlines into three strategic players: a Hub Major (American, Delta, or United with hub presence), Non-Hub Majors (major airlines without hub presence), and Non-Majors (low-cost carriers and regional airlines).

We test the null hypothesis that the true information structure is at least as informative as a \emph{privileged information} baseline, under which the hub airline observes the complete vector of payoff shocks while competitors observe only their own payoff realizations. Our results provide strong evidence against the privileged information hypothesis. Under the conservative Bonferroni correction, we reject the null in 2 of 16 market types. Under the less conservative Holm procedure, we reject in 7 market types. The pattern of rejections is informative: all rejections occur in markets characterized by large size or high hub airline presence. These findings suggest that hub airlines' entry decisions are based on less precise information about competitors than the privileged structure assumes, consistent with an information structure closer to incomplete information where each player observes primarily their own payoff shock. Our results have implications for counterfactual policy analysis, suggesting that researchers should be cautious about assuming that hub airlines have superior information when simulating entry behavior under alternative market conditions.

Beyond our application, we highlight the generality and potentially applicability of our method in two areas. First, one of the main practical hurdles when performing estimation under weak assumptions on information is that counterfactual prediction may be uninformative. Our testing approach may help produce sharper counterfactuals: when our test rejects null hypotheses in favor of a more informative information structure, which can then be used in counterfactuals, the researcher obtains a tighter set of predictions. Second, there is a range of strategic empirical contexts where our test can be applied to learn itneresting aspects of the information structure. For instance, \cite{magnolfi2016political} study entry in the Italian supermarket industry, where one player has political connections that can both affect payoffs directly and provide information about rivals' costs. Similarly, \cite{baltrunaite2020political} studies procurement auctions in which firms can buy preferential treatment via donations to politicians, with information on competitors' valuations or bids serving as the channel. Our test can be applied to such environments to detect the presence of political connections or corruption.

\subsection*{Related Literature}

Our paper contributes to the literature on econometric analysis of games \citep[see][for surveys]{dePaula2013ARE,Aradillas-Lopez:2020aa}. Within this literature, there has been growing attention to the role of assumptions on information. Earlier work typically assumed a specific information structure---either complete information, where all players observe the full vector of payoff shocks \citep[e.g.,][]{bresnahan1991empirical,Tam03,ciliberto2009market}, or incomplete information, where each player observes only their own payoff type \citep[e.g.,][]{seim2006empirical,Bajari:2007aa,sweeting2009strategic}. Since the predictions of these models generally differ, the choice of information structure is consequential for estimation and counterfactual analysis.

More recent work has proposed frameworks that relax strong assumptions on information. \cite{Grieco:2014aa} introduces a flexible parametric information structure that nests complete and incomplete information as special cases, enabling estimation and testing within a unified framework. A related strand of literature adopts Bayes Correlated Equilibrium (BCE) as a solution concept, which summarizes the behavioral implications of all information structures that are expansions of a specified baseline \citep{Bergemann:2016aa}. \cite{magnolfi2023estimation} show that BCE provides a tractable framework for estimation of discrete games when the researcher is agnostic about the true information structure, and apply their approach to radio advertising. \cite{Syrgkanis:2017aa} employ BCE in the context of common-value and private-value auctions. \cite{gualdani2024identification} study identification in single-agent discrete choice problems under BCE. In recent work, \cite{koh2023stable} introduces \emph{Bayes stable equilibrium} (BSE), which refines BCE by requiring that players have no incentives to deviate after observing opponents' realized actions, and estimates an entry model under this assumption.\footnote{Other studies use BCE (or related notions) in empirical work not for informational robustness but for other purposes, such as counterfactual prediction \citep{canen2023synthetic} or capturing strategic play by AI agents \citep{lomys2024structuralAI}.} While the existing literature focuses on estimation of payoff parameters under weak informational assumptions, our paper instead focuses on testing hypotheses about the information structure itself without relying on a parametric approach.

Our approach leverages the robust implications of an incomplete model that admits multiple equilibria with unknown information structures. The implied restrictions exploit the convexity of the BCE prediction and take the form of moment inequalities. There have been important developments on ways to systematically derive sharp identifying restrictions in models with set-valued predictions without making assumptions on unknown equilibrium selection mechanisms \citep{Tam03,beresteanu2011sharp,galichon2011set}. We utilize the support function of the BCE prediction to construct a test statistic, following the approach in \cite{beresteanu2011sharp} for models with convex predictions. We propose a new way to studentize the sample moments used to calculate the statistic. A well-known issue with inference based on moment inequalities is that test statistics may have different asymptotic distributions depending on the configuration of the data-generating process \citep[see][for a summary]{Canay_Shaikh_2017}. We adopt a moment selection procedure \citep{Andrews:2010aa} to ensure that the test is asymptotically valid uniformly across a wide range of data-generating processes.

Our empirical application contributes to the literature on competition and market structure in the airline industry  \citep{berry1992estimation,ciliberto2009market,ciliberto2021market}. An airline's profitability depends crucially on which competitors choose to serve the same route. \cite{Berry:1996aa} documents that hub airlines benefit from both demand-side advantages (through improved service quality for time-sensitive travelers) and cost-side economies stemming from increased flight frequency and operational efficiencies. Our analysis abstracts from the full network structure of airline operations (as modeled in \cite{bontemps2023online} and \cite{yuan2024network}) to focus specifically on whether hub status confers informational advantages beyond these well-documented cost and demand benefits.

The remainder of the paper is organized as follows. Section \ref{sec:model} describes the model of an empirical discrete game, including payoff structure and information structure. Section \ref{sec:comparing_info} introduces methods for comparing information structures via expansion and individual sufficiency orderings. Section \ref{sec:test} develops our testing framework, including the test statistic, bootstrap procedure, and asymptotic properties, while Section \ref{sec:extensions} contains extensions to our procedure. Section \ref{sec:simulation} presents Monte Carlo evidence on the finite-sample performance of the test. Section \ref{sec:application} applies our methodology to the U.S.\ airline industry. Section \ref{sec:conclusion} concludes.

\onehalfspacing
\section{Model}\label{sec:model}
We describe a general model of an empirical discrete game, the environment where we develop our testing procedure. This model is similar to the one described in \cite{magnolfi2023estimation}. Games in the class we describe are indexed by realizations of covariates $x\in X$. Players are indexed in a finite set $N\equiv \{1,...,|N|\}$. Each player $i\in N$ chooses an action $y_{i}\in Y_{i}$, a discrete set. Both the actions' set $Y\equiv\times_{i\in N}Y_{i}$ and $N$ do not depend on $x$. All these aspects of the game are common knowledge among players and known to the researcher. The researcher jointly observes actions $y\in Y$ and covariates $x$. 
Next, we describe payoff structure and information structure of the game.

\subsection{Payoff Structure and Information Structure}

Player $i$ is characterized by a  payoff type $\varepsilon_{i}\in\E_{i}$. The vector of payoff types $\varepsilon\equiv(\varepsilon_{i})_{i\in N}$ is distributed according to the cumulative distribution function (CDF) $F\left(\cdot;\theta_{\varepsilon}\right)$, which is known to the researcher up to the finite-dimensional vector $\theta_{\varepsilon}$. Payoffs to player $i$, are denoted by $\pi_{i},$ and are realized according to the function $\pi_{i}(\cdot,\cdot;x,\theta_{\pi}):Y\times\E_{i}\rightarrow\mathbb{R}.$ We assume that payoff types $\varepsilon$ are independent of covariates $x.$ A realization of $x$ and a vector of parameters $\theta\equiv\left(\theta_{\pi},\theta_{\varepsilon}\right)\in\Theta$ fully characterize the payoff structure of the game.

We assume that every player $i$ knows the parameters and $x$. 
Each player also observes a private random signal $\tau_{i}^{x} \in \mathcal{T}_i$, which may carry information on the vector of payoff types $\varepsilon.$ Let $\mathcal T=\times_{i\in N}\mathcal T_i$ be the space of signals equipped with a $\sigma$-algebra $\mathcal F$.\footnote{In what follows, we tacitly assume the signals are defined on the measurable space $(\mathcal T,\mathcal F)$ but omit the underlying $\sigma$-algebra from the definition of information structure for notational simplicity.}

We define a (generic) \emph{information structure} $S$ as a mapping from covariate values to conditional distributions of signals:
\[
S:\; x \mapsto \Bigl(\mathcal T_x,\{P^{x}_{\tau\mid \varepsilon}:\varepsilon\in\E\}\Bigr),
\]
where $\{P^{x}_{\tau\mid \varepsilon}:\varepsilon\in\E\}$ is a probability kernel governing the distribution of the signal vector $\tau^{x}$ conditional on each realization of $\varepsilon$. The set $\mathcal T_x\subset\mathcal T$ denotes the support of these conditional distributions.

Allowing $S$ to depend on $x$ reflects the possibility that the informational environment varies with observable covariates. In what follows, we write $\game{S}$ for the game with covariates $x$, payoff parameters $\theta$, and information structure $S(x)$.\\

\begin{example}\rm
\label{ex:entry_game} Consider the entry game first proposed in  \cite{bresnahan1991empirical}, where players are two firms that are potential entrants in a market, and choose to either ``Enter'' or ``Not enter'', corresponding to $y_{i}=1$ and $y_{i}=0$, respectively. The researcher observes entry choices across a set of markets with covariates $x.$ Firm $i$'s profits are zero when not entering, and $\pi_{i}(y,\varepsilon_{i};x,\theta_{\pi})=x^{\prime}\beta + \Delta y_{-i}+\varepsilon_{i}$ upon entry.
The payoff types (e.g., cost shifters) $\varepsilon_{i}$ are unobservable to the researcher;  an information structure $S$ specifies the information that player $i$ has on its opponent's $\varepsilon_{-i}$. In addition to knowing its own payoff type, firm $i$ may have access to a noisy measurement  of the opponent firms' unobserved cost shifters $\varepsilon_{-i}\in \mathbb R^{|N|-1}$. For example, a firm with a larger market presence may have more information about market profitability and costs faced by other firms.  One may model the signal as a random vector $\tau^x_i\in\mathcal T^x=\mathbb R^{|N|}$ following an unknown conditional distribution $P^x_{\tau|\varepsilon}$.
\end{example}

\begin{table}[htbp]
\centering
\begin{threeparttable}
\caption{Examples of Information Structure}
\label{tab:ex_info_structure}
\small
\begin{tabular}{lll}
\hline\hline
Structure & Notation & Decription \\
\hline
Null information &
$S_{\text{Null}}$ &
$\tau_i \perp \varepsilon$ for all $i$ \\[0.6em]

Incomplete information &
$S_I$ &
$\tau_i$ reveals $\varepsilon_i$; partially informative on $\varepsilon_{-i}$ \\[0.6em]

Complete information &
$S_C$ &
$\tau_i$ fully reveals $\varepsilon$ \\[0.6em]

Privileged information &
$S_P$ &
$\tau_1$ reveals $(\varepsilon_1,\varepsilon_2)$; $\tau_2$ reveals $\varepsilon_2$ only \\[0.6em]

Public signals &
$S_{\text{Pub}}$ &
$\tau_i$ reveals $\nu_{-i}$ for all $i$, where $\varepsilon_i=\nu_i+\epsilon_i$ \\
\hline\hline
\end{tabular}

\begin{tablenotes}[flushleft]
\small
\item \textit{Notes:} Dependence of signal kernels on covariates $x$ is omitted for notational simplicity. The privileged information structure is defined for a two-player game.
\end{tablenotes}
\end{threeparttable}
\end{table}

We introduce here some special cases of information structure that are useful in what follows (summarized in Table \ref{tab:ex_info_structure}).  We call the \textit{null information} structure $S_{Null}$ the one characterized by fully uninformative signals, so that each $\tau_i$ does not affect players' beliefs on the realization of $\varepsilon$. In terms of the conditional law of the signal, $P^x_{\tau|\varepsilon}$ corresponds to the setting where $\tau$ is independent of $\varepsilon$ for all $x$.
We define the \textit{incomplete information} structure $S_I$ as the one characterized by signals $\tau_i$ that reveal the realization of $\varepsilon_i,$ but are only partially informative on  $\varepsilon_{-i}$. Coupled with the assumption of independent types, this assumption is adopted in seminal work on social interaction and econometrics of games \citep[e.g., ][]{brock2001discrete,seim2006empirical}. 
Conversely, in a game with the \textit{complete information} structure $S_C$, the signal $\tau_i$ is fully informative on the vector $\varepsilon$ \citep[as in e.g.,][]{bresnahan1991empirical,Tam03}. Next, consider the \textit{privileged information} structure $S_P$. To simplify the discussion, we define this information structure in the context of the two players game of Example 1; extensions to more general games are immediate. In $S_P$ player 1 has more information than player 2. In this case,  $\tau_2$ is only informative on  $\varepsilon_{2}$, whereas $\tau_1$ fully reveals  $\varepsilon=(\varepsilon_{1},\varepsilon_{2})$. Finally, suppose that payoff types for each player $i$ are $\varepsilon_i = \nu_i + \epsilon_i$. We define the \textit{public signals} information structure $S_{Pub}$ as one where signals $\tau_i$ reveal the opponent's shock $\nu_{-i}$ symmetrically for each player.

Notice that, for notational simplicity, in describing the information structures above we omit the dependence on $x$ of signal kernels and support sets. An information structure $S(x)$ could feature, for different values of $x$, any of the special information structures described above.

\subsection{Equilibrium and Predictions}

 We define a notion of equilibrium in our model, namely the Bayes-Nash equilibrium (BNE), that links the primitives of game $\Gamma^{x}(\theta,S)$ to players' actions. Unlike in most of the previous literature where the researcher specifies $S$ to utilize the BNE, this paper's approach is to test for $S$ assuming that actions are generated from the BNE.
 
\begin{definition}[Bayes Nash Equilibrium] Let $S$ be an information structure. A strategy profile $\sigma=\left(\sigma_{1},\dots,\sigma_{|N|}\right)$ is a Bayes Nash Equilibrium (BNE) of the game $\Gamma^{x}(\theta,S)$ if for every $i\in N,$ $\varepsilon_{i}\in\E_{i}$ and $\tau_{i}\in \mathcal{T}_{i}^{x}$ we have that, whenever for some $y_{i}\in Y_{i}$ the corresponding $\sigma_{i}\left(y_{i}\mid\varepsilon_{i},t_{i}\right)>0,$ then :
\begin{align*}
E_{\sigma_{-i}}\left[\pi_{i}\left(y_{i},y_{-i},\varepsilon_{i};x,\theta_{\pi}\right)\mid\varepsilon_{i},\tau_{i}\right] & \geq E_{\sigma_{-i}}\left[\pi_{i}\left(y_{i}^{\prime},y_{-i},\varepsilon_{i};x,\theta_{\pi}\right)\mid\varepsilon_{i},\tau_{i}\right],\quad\forall\:y_{i}^{\prime}\in Y_{i},
\end{align*}
where the expectation of $y_{-i}$ is taken with respect to the distribution of equilibrium play $\sigma_{-i}\left(y_{-i}\mid\varepsilon_{j},\tau_{j}\right)=\prod_{j\neq i}\sigma_{j}\left(y_{j}\mid\varepsilon_{j},\tau_{j}\right).$
\end{definition}
Let the set $\mathit{BNE^{x}}(\theta,S)$ be the set of Bayes-Nash equilibria for game $\game{S}$. The following set collects all conditional outcome distributions compatible with some BNE in the game. 
\begin{definition}[BNE Predictions] For a game $\game{S},$
the set of \emph{BNE predictions} is 
\[
\bne{S}=\Big\{q(\cdot|x)\in\Delta^{|Y|}: q\left(y|x\right)=E[\sigma(y|\varepsilon,\tau)|x],\sigma\in\mathit{BNE^{x}}(\theta,S)\Big\},
\]
where expectation $E[\cdot|x]$ is with respect to the conditional distribution of $(\varepsilon,\tau)$ determined by the signal distribution 
 $P^x_{\tau|\varepsilon}$ and prior $F(\cdot;\theta_\varepsilon)$.
\end{definition}
We assume that players play a Bayesian Nash equilibrium under an unknown information structure $S$ (see Assumption~\ref{as:Model} below). Precisely characterizing the resulting BNE predictions is challenging for two reasons. First, it requires specifying the underlying information structure $S$. Second, even given $S$, computing $\bne{S}$ entails characterizing all fixed points of the best-response correspondence. These requirements make it difficult to test restrictions on $S$ directly using $\bne{S}$.

To address these challenges, below in Section \ref{sec:test}, we introduce an alternative set of predictions based on the notion of \emph{Bayes correlated equilibrium}. This concept serves as a useful device for two reasons. First, it is compatible with any BNE outcome generated under an information structure that is at least as informative as a given baseline. Second, the resulting prediction set admits a computationally tractable characterization.

\section{Comparing Information Structures}\label{sec:comparing_info}
To make progress towards our goal of using the data to test hypotheses on the information structure of the game, we need to be able to have a rigorous way of comparing information structures. This is a complex task. Fix a value of $x$ and consider, for instance, two information structures for the same collection of games:
\begin{eqnarray*}
\es{x}{1} & = & \bigl(\mathcal{T}^{1,x},\{ P_{\tau|\varepsilon}^{1,x}:\:\varepsilon\in\E\} \bigr)\\
\es{x}{2} & = & \bigl(\mathcal{T}^{2,x},\{ P_{\tau|\varepsilon}^{2,x}:\:\varepsilon\in\E\} \bigr).
\end{eqnarray*}
Each information structure can be a complex, high-dimensional object, so that it is not immediate to compare $\es{x}{1}$ and $\es{x}{2}$. More precisely, we would want to form an \emph{ordering} among information structures, as to give a rigorous meaning to the statement ``$\es{x}{1}$ is more informative than $\es{x}{2}$.'' We introduce two interrelated notions of ordering in this subsection.

The first notion of ordering starts from a simple consideration: if there is an information structure that contains all the information present in $\es{x}{1}$ and $\es{x}{2}$, then this new combined information structure is clearly at least as informative as either $\es{x}{1}$ or $\es{x}{2}$. We thus follow \cite{Bergemann:2016aa} and define combinations of information structures. For this, we assume the space $\mathcal T$ is finite.
\begin{definition}
[Combination] The information structure at $x$
\[
\es{x}{*}=\bigl(\mathcal{T}^{*,x},\left\{ P_{\tau^*|\varepsilon}^{x}:\:\varepsilon\in\E\right\} \bigr)
\]
is a combination of $\es{x}{1}$ and $\es{x}{2}$ if 
\begin{eqnarray*}
\mathcal{T}^{*,x}&=& \prod_{i=1}^{|N|} \mathcal{T}_{i}^{*,x}\\
\mathcal{T}_{i}^{*,x} & = & \mathcal{T}_{i}^{1,x}\times \mathcal{T}_{i}^{2,x},\quad i \in N.\\
\sum_{\tau^{2}}P_{\tau^{*}|\varepsilon}^{x}(\tau^{1},\tau^{2}|\varepsilon) & = & P_{\tau^{1}|\varepsilon}^{x}(\tau^{1}|\varepsilon)\quad\text{for each \ensuremath{\tau^{1}}},\\
\sum_{\tau^{1}}P_{\tau^{*}|\varepsilon}^{x}(\tau^{1},\tau^{2}|\varepsilon) & = & P_{\tau^{2}|\varepsilon}^{x}(\tau^{2}|\varepsilon)\quad\text{for each \ensuremath{\tau^{2}}}.
\end{eqnarray*}
\end{definition}

Intuitively, the combined information structure $\es{x}{*}$ gives players access to both $\es{x}{1}$ and $\es{x}{2}$. Based on the definition of combination, we can also define, for any information structure $\es{x}{1},$ its \emph{expansions}. 
\begin{definition}[Expansion]
An information structure $\es{x}{*}$ is an expansion of $\es{x}{1}$ if it is a combination of $\es{x}{1}$ and $\es{x}{2}$ for some $\es{x}{2}$. We write $S^{*}\ex S^{1}$ if $S^{*}(x)$ is an expansion for $S^{1}(x)$ for all $x\in X$.
\end{definition}
Clearly, if $\es{x}{*}$ is an expansion of $\es{x}{1},$ it is more informative than $\es{x}{1}$ for the players that observe its signals. 
We consider settings in which players of a game are believed to have access to at least a baseline information structure. For this purpose, we define a set of information structures that are at least as informative as some baseline.
\begin{definition}[Set of Expansions]
For any $S_B$, $\cS(S_B)$ is the set of all information structures that are expansions of $S_B$.   
\end{definition}

When constructing $\cS(S_B)$, we refer to $S_B$ as the \textit{baseline} information structure for the set---in fact, any information structure in $\cS(S_B)$ will be at least as informative as $S_B$ in the sense of expansions. Moreover, notice that in our context, every information structure is an expansion of the null information $S_{\mathrm{Null}}$. Thus, we define $\cS = \cS(S_{\mathrm{Null}})$ as the universe of information structures that we consider. In the context of specification testing, we consider strict subsets of $\cS$ as potential baseline information structures.

We also consider a second notion of ordering of information structures to make further progress. 
\cite{Bergemann:2016aa} present the following intuitive way to (partially) order information structures:
\begin{definition}
[Individual Sufficiency] $\es{x}{1}$ is individually sufficient for $\es{x}{2}$ if there exist a combined information structure $\es{x}{*}$ such that, for each $i$ and measurable set $A\subset \mathcal T_i$,
\[
P^{*,x}(\tau_{i}^{2}\in A|\tau_{i}^{1},\tau_{-i}^{1},\varepsilon)=P^{*,x}(\tau_{i}^{2}\in A|\tau_{i}^{1}),~a.s.
\]
where the probability of $P^{*,x}(\tau_{i}^{2}|\tau_{i}^{1},\tau_{-i}^{1},\varepsilon)$ is computed using the combined kernel $P_{\tau^{*}|\varepsilon}^{*,x}$. 
\end{definition}
In short, $\es{x}{1}$ is individually sufficient for $\es{x}{2}$ if there is a combined information structure of the two, in which, for each $i$,
\begin{align}
    \tau_{i}^{2} \perp (\tau_{-i}^{1},\varepsilon)|\tau_{i}^{1}.
\end{align}
Hence, a player, given his signal $\tau_{i}^{1}$, and given he knows the information structures $\es{x}{1}$ and $\es{x}{2},$ is able to compute what signal he would have received according to the information structure $\es{x}{2}.$ The individual sufficiency property reduces to \emph{Blackwell sufficiency} in the one-player case. Clearly, the ordering of information structures based on individual sufficiency is a partial one. Let us introduce the ordering of information structures based on the individual sufficiency. 
\begin{definition}[]
An information structure $S^{1}$ is individually sufficient for $S^{2}$ if  $S^{1}(x)$ is individually sufficient for $S^{2}(x)$ for all $x\in X$. We write $S^{1}\is S^{2}$ whenever $S^{1}$ is individually sufficient for $S^{2}$.
\end{definition}

\begin{remark}\rm
 The individual sufficiency ordering above is defined by comparing $S^1(x)$ and $S^2(x)$ for all $x \in X$. 
If instead one is interested in ordering information structures over a specific subset $\tilde X \subset X$, 
one can define the restricted ordering $S^1 \is_{\tilde X} S^2$ by requiring that $S^1(x)$ be individually sufficient for $S^2(x)$ for every $x \in \tilde X$.
\end{remark}

\noindent\textbf{Example 2:} Recall the \textit{public signals} information structure $S_{Pub}$ where signals $\tau_i$ reveal the opponent's shock $\nu_{-i}$ symmetrically for each player in payoff $\varepsilon_i = \nu_i + \epsilon_i$. Clearly, this information structure represents an expansion of $S_I$, and is individually sufficient for $S_I$ so that we can write $S_{Pub} \is S_I$. \\

There is a tight relation between the partial orders of information structures induced by expansion and individual sufficiency. In fact, these two concepts are interchangeable if we consider as the same all the information structures that have the same canonical representation \citep{mertens1985formulation}, or induce the same beliefs about the state. This assumption, which refines the space of all information structures, is natural in our context. Ultimately, we are interested in testing hypotheses on information using data, and it is not possible to distinguish empirically two information structures that induce the same beliefs, and hence the same actions. Hence, if we restrict our attention to information structures with different canonical representation, Claim 1 in \cite{Bergemann:2016aa} implies:
\begin{lemma}\label{lem:IS_exp}
For any baseline information structure $S_B$, we let the set of information structures $\cS^{*}(S_B)$ be the subset of $\cS(S_B)$ such that any two $S, S^{\prime}\in \cS^{*}(S_B)$ have a different canonical representation. Then, for any $S,S'\in \cS^*(S_B)$, $S\is S^{\prime}$ if and only if $S \ex S^{\prime}$. 
\end{lemma}

\noindent Lemma \ref{lem:IS_exp} defines the set $\cS^{*}(S_B)$ as the collection of all possible information structures $S$ that have a different canonical representation and are expansions of the baseline $S_B$. Within such sets $\cS^{*}(S_B)$, the partial orders implied by expansion and individual sufficiency are equivalent. We take $\cS^*\equiv\cS^*(S_{\mathrm{Null}})$ as the parameter space for the unknown information structure.\\

\section{Tests of Information Ordering}\label{sec:test}

We now ready to introduce the framework of hypothesis testing. Although specifying the game's true information structure $S\in \cS^*$ is difficult, the researcher may want to test that it is at least as informative as some \emph{baseline information structure} $S^r$. Formally, we consider testing
\begin{align}
H_{0}:  S\succeq S^{r}~~~\text{vs}~~~
H_{1}:  S\nsucceq S^{r}.\label{eq:h0h1} 
\end{align}
Note that \eqref{eq:h0h1} is a one-sided hypothesis regarding the unknown information structure $S$.

For example, one may test whether one of the players has privileged information in specific markets by choosing $S^r$ properly, as in Example 3. An important aspect of our environment is its generality, which allows for the possibility of accommodating and testing information structures that vary across markets. As $S$ specifies an information structure for every game $x$, the distribution of signals $P_{\tau|\varepsilon}^{x}$ may vary across markets $x$ for a fixed $S$ in the DGP. Consider now a restriction $S^{r}$, testing a statement such as $S\is S^{r}$: this restricts the heterogeneity of information structures across markets, as we are implicitly imposing $S(x)\succeq S^{r}$ for all $x$. In this case, the restriction in $S^{r}$ is the same across all markets as in previous examples (e.g., complete information for every $x$). However, $S^{r}$ may also be heterogeneous across markets, in which case we test the null hypothesis $H_{0,x}:  S(x)\succeq S^{r}(x)$ for given $x\in X$; see Section \ref{subsec:CS_markets} for this extension. Examples 3 and 4 below consider this type of null hypothesis.
\\

\noindent \textbf{Example 3:}  Consider the empirical setting of entry in airline markets similar to  \cite{ciliberto2009market}, where each market $m$ is a city-pair, and airlines are players in a binary entry game. As in our application of Section \ref{sec:application}, we focus on markets where at least one endpoint is a hub for a major airline, and we consider three separate players in each market: the \emph{Hub Major} player represents major airlines with a hub at least one endpoint,  the \emph{Non-Hub Major}, and \emph{Non-Major} player (which includes low-cost carriers). 
Moreover, we include in the vector $x_{m}$ airline-specific shifters of profitability such as measures of market presence and market size. Although $S(x)$ may depend on  $x_{m}$, we specify a restriction on information $S^r$ that does not depend on covariates. Specifically, for each market $m$, $S^r$ prescribes that the hub airline player has complete information, whereas the other two players have incomplete information, receiving uninformative signals. In essence, this specification of $S^r$ corresponds to the privileged information structure $S_{P}$ defined in Section 2.1, and allows the researcher to test whether hubs confer informational advantages over competitors' cost and profit. 

\noindent\textbf{Example 4:} Consider the empirical setting of entry in grocery markets similar to  \cite{magnolfi2016political}, where supermarket firms are potential entrants in each geographic market $m$. Because firms may have political connections that affect their entry costs, the vector of market and firm-specific covariates includes $x_{i,m}^{Conn},$ a measure of the intensity of the connection. In addition, political connections may also affect the information structure of the game, e.g., by providing connected players with superior information about rivals' costs. More concretely, we could test a restriction on information $S^r(x)$ such that players whose $x_{i,m}^{Conn}$ exceeds a certain threshold have complete information, whereas their competitors with low $x_{i,m}^{Conn}$ have incomplete information. In this framework, the identity of connected players may vary across markets according to  $x_{i,m}^{Conn}$, and more than one player may be connected in each market.\\

We note that our formulation of null and alternative in Equation \eqref{eq:h0h1} is well suited for some purposes (e.g., it makes sense to test the null hypothesis of a restriction on information if a researcher wants to analyze counterfactuals based on that restriction), but may not be appropriate for others. For instance, in Example 4 above, we would not impose the restriction of corruption or political connections as the null if the results are to be used to support an investigation into criminal behavior. Other formulations of the null are possible, and our framework may be extended to accommodate them.

\subsection{Equilibrium Concept and Predictions Robust to Information Specification}
Let us investigate the empirical contents of the null hypothesis. Given the hypothesis \eqref{eq:h0h1}, is it then possible to link the ordering of information structures to predictions in a way that enables researchers to test it? The answer is affirmative. For this,
we employ a solution concept whose predictions respect the ordering of information structures. This alternative solution concept is called \emph{Bayes correlated equilibrium} (BCE).\footnote{This definition is slightly different than the one in \cite{Bergemann:2016aa}. They define a decision rule to be a mapping
\[
\sigma:\mathcal{T}\times{\cal E}\rightarrow\Delta(Y)
\]
 So a decision rule is a kernel that describes the probability of every action profile for a given vector of payoff types/signals $\left(\varepsilon,\tau\right)$.
Hence, in \cite{Bergemann:2016aa}, a BCE is a decision rule that satisfies obedience. To define obedience, they combine the decision rule with the common prior defined on payoff types $\varepsilon$ and the prior distribution over signals $\tau$. In their definition, \cite{Bergemann:2016aa} do not need to explicitly impose consistency as an equilibrium requirement (it is implicit in their setting). Instead, we define a BCE to be a distribution over actions, payoff types and types, and require consistency to hold.} This equilibrium concept describes the implications for observables while staying agnostic about the underlying information structure. This equilibrium concept is the basis of our test statistic construction. Working with the BCE prediction is attractive because it respects the ordering of the information structures via set inclusion relationships. This is not necessarily the case if we were to work directly with the BNE prediction. In what follows, we let $\mathcal P_{W}$ denote the space of probability distributions defined on the underlying space $W$. For example,  $\mathcal{P}_{Y\times {\cal E}\times \mathcal{T}}$ represents the set of joint distributions of $(y,\epsilon,\tau)$.
\begin{definition}[Bayes Correlated Equilibrium] A Bayes Correlated Equilibrium $\nu^x\in\mathcal{P}_{Y\times {\cal E}\times \mathcal{T}}$ for the game $\game{S}$ is a probability measure
$\nu$ over action profiles, payoff types, and signals, which is
\begin{enumerate}
 \item \emph{Consistent with the prior};
\begin{align*}
\int_{A}\int_{Y}\nu^x(dy,d\varepsilon,dt)=\int_AP^x_{\tau|\varepsilon}(dt|\varepsilon)F(d\varepsilon;\theta_\varepsilon),
\end{align*}
 for any measurable $A\subseteq \mathcal E\times \mathcal T$; and
\item \emph{Incentive compatible}; for all $i,\varepsilon_{i},\tau_{i},y_{i}\text{ such that }\nu^{x}\left(y_{i}\mid\varepsilon_{i},\tau_{i}\right)>0,$
\begin{equation*}
E_{\nu}\left[\pi_{i}\left(y_{i},y_{-i},\varepsilon_{i};x,\theta_{\pi}\right)\mid y_{i},\varepsilon_{i},\tau_{i}\right]\geq E_{\nu}\left[\pi_{i}\left(y'_{i},y_{-i},\varepsilon_{i};x,\theta_{\pi}\right)\mid y_{i},\varepsilon_{i},\tau_{i}\right],\qquad\forall y_{i}^{\prime}\in Y_{i},
\end{equation*}
where the expectation operator $E_{\nu}\left[\cdot\mid y_{i},\varepsilon_{i},\tau_{i}\right]$
is taken with respect to the conditional equilibrium distribution
$\nu^x\left(y_{-i},\varepsilon_{-i},\tau_{-i}\mid y_{i},\varepsilon_{i},\tau_{i}\right).$
\end{enumerate}
\end{definition}
We interpret a Bayes correlated equilibrium as follows. Individuals are assumed to play a BNE under an unknown information structure $S$. From the analyst’s perspective, the resulting behavior admits the following equivalent description.
\begin{enumerate}
    \item There is a baseline information structure $S^{r}$. Players may possess additional information beyond that contained in $S^{r}$.
    \item A \emph{mediator} observes $\varepsilon \sim F(\cdot;\theta_{\pi})$ and a signal realization $\tau \sim P^{x}_{\tau\mid\varepsilon}$ generated under $S^{r}$.
    \item Conditional on $(\tau,\varepsilon)$, the mediator draws an action profile $y \sim \nu^{x}(\cdot\mid\tau,\varepsilon)$ and privately recommends action $y_i$ to each player $i$.
    \item Players \emph{obey} the mediator’s recommendations.
\end{enumerate}
This representation is convenient because the analyst need not specify the precise form of the true information structure $S$ as long as $S$ is at least as informative as the baseline structure $S^{r}$. In particular, as shown in Lemma 2 below, for every BNE outcome induced by some information structure $S \succeq S^{r}$, there exists a corresponding BCE outcome constructed under $S^{r}$ that generates the same distribution of observables.

Let $\mathit{BCE^{x}}(\theta,S)\subset \mathcal{P}_{Y,{\cal E},\mathcal{T}}$ denote the set of all BCE distributions for the game $\game{S}.$ A BCE distribution is a complex object describing the joint distribution of actions and payoff types. For any such distribution $\nu$ we define a compatible  \emph{BCE prediction} $\bce{S}$. This set collects the distributions of the observables that are compatible with BCE \citep{magnolfi2023estimation}. Formally:
\begin{definition}[BCE Predictions] For a game $\game{S},$
the set of \emph{BCE predictions} is:
\[
\bce{S}=\left\{q\in \Delta^{|Y|} :q(y)=\int_{\mathcal E\times\mathcal T} \nu^x\left(y,d\varepsilon,d\tau\right),\nu^x\in\mathit{BCE^{x}}(\theta,S)\right\} .
\]
\end{definition}
An advantage of working with the BCE prediction is that $\bce{S}$ is a compact convex set because the incentive compatibility and consistency with prior impose linear restrictions on $\nu$. We will use this feature to construct a computationally tractable test statistic.
The robust prediction result of \cite{Bergemann:2016aa}, which we state here for ease of reference, links the BCE and BNE predictions of the game as follows:
\begin{prop}\label{lem:bce_bne}
For all $\theta\in\Theta$ and $x\in X$, 
\begin{enumerate}
\item [(i)]If $q\in Q_{\theta, S^{r}}^{\mathit{BCE}}\left(x\right),$ then $q\in \bne{S}$
for some $S\in\cS(S^{r}).$
\item [(ii)]Conversely, for all $S\in\cS(S^{r})$, $\bne{S}\subseteq Q_{\theta, S^{r}}^{\mathit{BCE}}\left(x\right).$
\end{enumerate}
\end{prop}
Additionally, there is a property of the ordering induced by individual sufficiency, which has implications for the observables:
\begin{prop}\label{prop:bce_bce}
$\es{x}{1}$ is individually sufficient for $\es{x}{2}$  if and only if $Q_{\theta,S^{1}}^{\mathit{BCE}}\left(x\right)\subseteq Q_{\theta,S^{2}}^{\mathit{BCE}}\left(x\right)$.
\end{prop}
This proposition is an extension of Theorem 2 in \cite{Bergemann:2016aa}, and spells out the consequences of individual sufficiency for players' behavior: if an information structure $\es{x}{1}$ is more informative (in the sense of individual sufficiency) of an information structure $\es{x}{2}$, the set of predictions corresponding to $\es{x}{1}$ is nested in the set of predictions implied by $\es{x}{2}$.

Propositions \ref{lem:bce_bne} and \ref{prop:bce_bce} allow comparisons of information structures based entirely on the sets of BCE and BNE predictions. We discuss more in depth the implications of our results for testing in the next section.

\subsection{Testable Implications}\label{ssec:test}

We now develop a test statistic for the hypothesis \eqref{eq:h0h1}. We make the following assumptions on the parameter space and data generating process.

\begin{assumption}\label{as:primitives} (i) $X$ is a finite set;
(ii) $\Theta\subset\mathbb{R}^{d_{\theta}}$ is a compact set; (iii)  $\cS^*$ is equipped with the partial order $\is$.
 \end{assumption}
Compactness of $\Theta$ is standard.  The parameter space $\cS^*$ for the information structure is partially ordered by the individual sufficiency ordering $\is$. For simplicity, we assume that the covariates are discrete (or discretized).\footnote{Allowing continuous covariates yields a continuum of conditional moment restrictions \citep[see e.g.][]{Chernozhukov:2013aa}. We leave this extension elsewhere to keep our tight focus on testing information structures.} 

The next assumption states, in each market, the data are generated from a Bayes-Nash equilibrium under an unknown information structure $S$.
\begin{assumption}\label{as:Model} The conditional choice probability $P_{y|x}$ satisfies
\begin{align}
P_{y|x}\in \bne{S},
\end{align}
for some $S\in \cS^*$ and $\theta\in\Theta.$ 
\end{assumption}
 A common empirical practice is to specify the true information structure $S$ and estimate $\theta$ using the restriction $P_{y|x}\in \bne{S}$.\footnote{Existing literature typically adopts as $S$  either the complete information structure \citep[e.g.,][]{bresnahan1991empirical,Tam03,ciliberto2009market}, or the incomplete information structure \citep[e.g.,][]{seim2006empirical,Bajari:2007aa,sweeting2009strategic}. } Instead, we specify the baseline information structure $S^r\in\cS^*$ and test \eqref{eq:h0h1}.
 
Let $P_x$ be the marginal law of $x$ and
let $P$ be the joint law of $(y,x)$.
For each $x$, we identify $P_{y|x}$ with a vector in the $|Y|-1$-dimensional simplex. We assume a sample is drawn from $P$ independently across markets.

\begin{assumption}\label{as:sampling}  $(y^{n},x^{n})=(y_\ell,x_\ell)_{\ell=1}^n$ is a
random sample from the law $P\in\mathcal P_{Y\times X}$. \end{assumption}
 If $H_0$ is true, the the following testable implication holds by the assumptions above and Proposition \ref{lem:bce_bne}(ii).
 \begin{lemma}
Suppose Assumptions \ref{as:primitives}-\ref{as:Model} hold. Suppose $S\is S^r$. Then, 
\begin{align}
	P_{y|x}\in\bne{S}\subseteq \bce{S^r},~\forall x \in X.\label{eq:test_implication}
\end{align}
\end{lemma}

Recall that $\bne{S}$ is a mixture over Bayesian Nash equilibria. When the true information structure $S$ is unknown, it is generally difficult to assess whether $S$ is more informative than $S^{r}$ by directly characterizing $\bne{S}$. Equation~\eqref{eq:test_implication} shows that this difficulty can be bypassed: it suffices to verify whether the conditional choice probability $P_{y\mid x}$ lies in a convex set $\bce{S^{r}}$.

The conditional law $P_{y\mid x}$ is directly identified from the data, and membership in the BCE prediction set $\bce{S^{r}}$ can be checked using computationally tractable procedures. Accordingly, the testable implication in \eqref{eq:test_implication} implies that the null hypothesis $H_{0}$ can be rejected whenever the conditional choice probability of some outcome falls outside the BCE prediction under $S^{r}$.

\setcounter{example}{0}
\begin{example}[revisited]\rm
Consider a simplified version of Example \ref{ex:entry_game} with two players, where the payoff function is given by $\pi_{i}\left(y,\varepsilon\right)=y_{i}\left(-\Delta_j y_{-i}+\varepsilon_{i}\right)$ for $(\Delta_1,\Delta_2)\in \Theta=(0,1]^2,j=1,2$, and $\varepsilon_i\stackrel{iid}{\sim} U[-1,1]$. 
Consider $S^r=S_C$, the complete information structure, as the baseline. 
According to the BCE prediction, the sharp lower bound for the probability of $Y=(1,0)$ is 
\begin{align}
LB_\Delta=\frac{1}{4}(1+\Delta_2(1-\Delta_1))\ge 0.25,~~~ \forall (\Delta_1,\Delta_2)\in \Theta.	
\end{align}

 Suppose the true information structure is $S_I$, where $\tau_i$ only reveals $\varepsilon_i$ but does not provide information about $\varepsilon_{-i}$.  A BNE under this information structure under the true interaction parameter $\Delta^*=(\Delta^*_1,\Delta^*_2)$ induces the following CCP:
\begin{align}
P(Y=(1,0))=\frac{1+\Delta^*_2}{(2+\Delta^*_1)(2+\Delta^*_2)}.	
\end{align}
For instance, if $\Delta^*_1=\Delta^*_2=0.5$, then $P(Y=(1,0))=0.24$, which is strictly below $LB_\Delta$  for any $\Delta$ in the parameter space.
Hence, we can detect the violation of $H_0$ by comparing the BCE prediction (i.e., $LB_\Delta$) and  the CCP. 
\end{example}

One may wonder why violations of the null hypothesis are detectable in the first place. The key intuition is that, under the postulated complete-information baseline, there exist values of $\varepsilon$ for which players know that playing $(1,0)$ is optimal. The BCE prediction therefore delivers a lower bound on the probability with which this outcome can occur.

By contrast, under incomplete information, supporting the action profile $(1,0)$ requires $\varepsilon_1$ to be sufficiently high (or $\varepsilon_2$ to be sufficiently low), since each player must optimally respond based solely on their own private signal $\varepsilon_i$. The resulting comparison depends on the threshold structure induced by $\Delta^{*}$. In the example above, the true conditional choice probability lies strictly outside the BCE prediction. For other choices of $\Delta^{*}$—for instance, $(\Delta^{*}_1,\Delta^{*}_2)=(0.2,0.8)$—the same conditional choice probability may instead fall within the BCE prediction for this outcome. Accordingly, our test detects a violation of the null hypothesis only when such discrepancies arise for at least one outcome and covariate value $x$.

\subsection{Test Statistics and Asymptotic Properties}
To fully exploit the testable implications of the model, we leverage all linear combinations of the conditional choice probabilities. 

For any $b\in\mathbb R^{|Y|}$ and a closed convex set $A$, let $h(b,A)=\sup_{a\in A}b^{\top}a$ be the \emph{support function} of $A$. 
We may restate \eqref{eq:test_implication} as follows: 
\begin{align}
    P_{y|x}\in \bce{S^r}~~\Leftrightarrow~~ b^{\top}P_{y|x}\le h(b,\bce{S^r}),~\forall b \in \ball ~\mathrm{and}~\forall x \in X,
\end{align}
where $\ball=\{b\in\mathbb R^{|Y|}:\wnorm{b}\le 1\}$ is the unit ball for a norm $\wnorm{b}=(b^{\top}W_{x}b)^{1/2}$, and $W_x$ is a positive definite matrix. 
Observe that we translate the hypothesis on the ordering of information structures into an ordering of functions, namely
$(b,x)\;\mapsto\; b^{\top} P_{y\mid x}
\quad \text{and} \quad
(b,x)\;\mapsto\; h\!\left(b,\bce{S^{r}}\right).$

For each $\theta$, let
\begin{align}
    T(\theta)\equiv\sup_{x\in X}\sup_{b\in \ball}\{b^{\top}P_{y|x}-h(b,Q^{BCE}_{\theta,S^r}(x))\}.\label{eq:defT}
\end{align}
Since $\ball$ contains the origin, $T(\theta)$ is nonnegative by construction. Moreover, under the null hypothesis, $T(\theta)=0$ for all $\theta$.\footnote{Indeed, if $b^{\top}P_{y\mid x}-h\bigl(b,Q^{BCE}_{\theta,S^{r}}(x)\bigr)\le 0$ for all $b\in\ball$, the supremum is attained at $b=0$.}

Our test is based on a sample analog of $T(\theta)$. Statistics built on support functions have been widely used in partially identified models \citep{beresteanu2008asymptotic,beresteanu2011sharp,bontemps2012set,kaido2014asymptotically}. While the construction of our test statistic is most closely related to \cite{beresteanu2011sharp}, it incorporates an important distinction: by introducing a weighted norm in the definition of $\ball$, we effectively studentize the moments used to form the test statistic.

To see this, let $\hat P_{n,x}$ denote the empirical conditional choice probability whose $j$-th component is $\hat P^{(j)}_{n,x}= \frac{1}{n_x}\sum_{\ell=1}^n1\{y_\ell=y^{(j)},x_\ell=x\},$
where $n_x=n^{-1}\sum_{\ell=1}^n1\{x_\ell=x\}$.  Let $W_{n,x}\equiv n\widehat{var}(\hat{P}_{n,x})$ be the sample covariance matrix of $\hat P_{n,x}$, and assume $W_{n,x}$ is positive definite. Let 
\begin{align}
    \balln\equiv\{b\in\mathbb R^{|Y|}:\wnormn{b}\le 1\},~\wnormn{b}\equiv(b^{\top}W_{n,x}b)^{1/2}.
\end{align}
Define the sample counterpart of $T(\theta)$:
\begin{align}
	T_{n}(\theta)&\equiv \sup_{x\in X}\sup_{b\in \balln}\sqrt n\{b^{\top}\hat P_{n,x}-h(b,\bce{S^r})\}.\label{eq:Tn1}
\end{align}     
This statistic can also be expressed as follows (see Lemma \ref{lem:studentization} in Appendix):
\begin{align}
 \sup_{b\in \balln}\sqrt n\{b^{\top}\hat P_{n,x}-h(b,\bce{S^r})\}=   \sup_{b\in \mathbb R^{|Y|}\setminus\{0\}}	\sqrt n\Big\{ \frac{b^{\top}\hat{P}_{n,x}-h(b,\bce{S^r})}{\sqrt{nb^{\top}W_{n,x}b}}\Big\}_+.\label{eq:studentization}
\end{align}
Hence, for each $b$, the quantity $\sqrt n\{b^{\top}\hat P_{n,x}-h(b,\bce{S^r})$ is divided by its standard error. This is desirable because $T_n$ constructed this way is less sensitive to imprecisely estimated moments than a statistic without any studentization. 

When $\hat P_{n,x}$ is outside the BCE prediction, 
the statistic measures the maximum deviation of the empirical distribution from the BCE prediction for each value of $\theta$. Using techniques in \cite{magnolfi2023estimation}, this statistic can be calculated by solving a convex program. We provide details on computational aspects in Section \ref{ssec:computation}.

The behavior of the test statistic $T_n$ depends on the location of $P_{y\mid x}$ relative to the BCE prediction. To make this dependence explicit, we rewrite $T_n$ as
\begin{align}
T_n(\theta)
=
\sup_{x\in X}\sup_{b\in \balln}
\bigl\{
\mathbb G_n(b,x)+\eta_{\theta}(b,x)
\bigr\},
\label{eq:Tn2}
\end{align}
where $\mathbb G_n(b,x)\equiv \sqrt{n}\, b^{\top}(\hat P_{n,x}-P_{y\mid x})$ and
\[
\eta_{\theta}(b,x)
\equiv
\sqrt{n}\bigl(b^{\top}P_{y\mid x}-h(b,\bce{S^{r}})\bigr).
\]

The values of $(b,x)$ that are relevant for the supremum in \eqref{eq:Tn2} are those for which $\eta_{\theta}(b,x)$ is close to zero. A key challenge is that $\eta_{\theta}(b,x)$ cannot be consistently estimated uniformly over data-generating processes. To address this issue, we combine bootstrap methods with a moment selection procedure to ensure uniform validity. 

Let
\[
\hat \eta_{n,\theta}(b,x)
\equiv
\sqrt{n}\bigl(b^{\top}\hat P_{n,x}-h(b,\bce{S^{r}})\bigr),
\]
and let $\{\tau_n\}$ be a slowly diverging sequence, such as $\tau_n=\sqrt{\log n}$. Define
\begin{align}
\Psi_{n,\theta}
\;\equiv\;
\bigl\{
(b,x)\in \mathbb{B}_{xn}\times X
:\;
\hat \eta_{n,\theta}(b,x)\ge -\tau_n
\bigr\}.
\end{align}
The set $\Psi_{n,\theta}$ collects those $(b,x)$ pairs for which the sample moment restrictions are nearly binding.
The intuition is as follows. When a moment restriction is strictly slack in the population, the corresponding statistic $\hat \eta_{n,\theta}(b,x)$ diverges to $-\infty$ at rate $\sqrt{n}$. Excluding such moments improves the accuracy of the bootstrap approximation. At the same time, to avoid mistakenly discarding moments that are close to binding, we adopt a conservative rule that retains moments whose slackness does not exceed the slowly diverging threshold $\tau_n$. This construction follows the insights of generalized moment selection procedures \citep[see e.g.][]{Andrews:2010aa,Chernozhukov:2013aa}.

Let $\{(y^{*n}_\ell,x^{*n}_\ell)\}_{\ell=1}^n$ be a bootstrap sample drawn from the empirical distribution $\hat P_n$,\footnote{The empirical distribution is $\hat P_n=\frac{1}{n}\sum_{\ell=1}^n \delta_{(y_\ell,x_\ell)}$, where $\delta_{(y,x)}$ denotes a point mass at $(y,x)$.} and let $\hat P^*_{n,x}\in\mathbb R^{|Y|}$ denote the vector of empirical conditional frequencies of outcomes in the bootstrap sample given $x^{*}=x$. Define
$\mathbb G^*_n(b,x)
\equiv
\sqrt{n}\, b^{\top}(\hat P^*_{n,x}-\hat P_{n,x}).$
A bootstrap analog of $T_n$ is then given by
\begin{align}
T^*_{n}(\theta)
\;\equiv\;
\sup_{(b,x)\in \Psi_{n,\theta}}
\bigl\{
\mathbb G^*_n(b,x)
\bigr\}.
\end{align}

For each $\theta\in\Theta$, let the bootstrap $p$-value be
\begin{align}\label{eq:p_value}
	p_n(\theta)\equiv P^*(T^*_{n}(\theta)>T_{n}(\theta)|y^n,x^n),~
\end{align}
where $P^*$ is the  distribution of $T^*_{n}$ conditional on the  sample $(y^n,x^n)$. 
We reject $H_0$ if
\begin{align}
    p_n(\theta)\le \alpha\quad\text{for all } \theta\in\Theta.
\end{align}
Otherwise, there exists a parameter value $\theta\in\Theta$ for which an information structure $S\is S^r$ is consistent with data. We may then collect such parameter values to define a confidence set:
\begin{align}
    CS_{n}\equiv\{\theta\in\Theta:p_n(\theta)>\alpha\}.
\end{align}
This set covers $\theta$ with probability $1-\alpha$ asymptotically. 
Theorem \ref{thm:size_control} below ensures that the test is asymptotically valid over a large class of data-generating processes. 

For any $M$-by-$M$ matrix $A$, let $\|A\|_{op}=\inf\{c\ge 0:\|Az\|\le c\|z\|,~\forall z\in \mathbb R^M\}$ be its operator norm. We impose the following regularity conditions.
\begin{assumption}
	\label{as:Wn}
	 There exists $W_x\in\mathbb R^{|Y|\times |Y|}$ with $x\in X$ such that, for any $\epsilon>0$,   $$P(\sup_{x\in X}\|W_{n,x}-W_x\|_{op}>\epsilon)\le \epsilon,~\forall n\ge N_\epsilon,$$
for some $N_\epsilon\in\mathbb N$;
(ii) $\underline{\kappa}<\underline{\lambda}(W_x)$ and $\overline{\lambda}(W_x)<\overline{\kappa}$ for  uniform constants $0<\underline{\kappa}\le \overline{\kappa}<\infty$ for all $x\in X$, where $\underline{\lambda}(\cdot)$ and $\overline{\lambda}(\cdot)$ are the smallest and largest eigenvalues, respectively.
\end{assumption}

\begin{assumption}
	\label{as:px}
	There is $ \underline{\zeta}>0$ such that $P_{y|x}(y_\ell=y|x)\ge \underline{\zeta}$  and $P(x_\ell=x)\ge \underline{\zeta}$ for all $y\in Y$ and $x\in X$.
\end{assumption}
Assumption \ref{as:Wn} states the sample weighting matrix converges uniformly to its limiting counterpart $W_x$. We also assume $W_x$'s eigenvalues are uniformly bounded away from 0 and from above. Assumption \ref{as:px} requires $y_i$'s conditional probability and $x_i$'s probability mass function are uniformly bounded away from 0.

For each $\theta\in\Theta$, we define our null model $\mathcal P_\theta$ as follows: 
\begin{align}
	\mathcal P_\theta\equiv\Big\{P\in\mathcal P_{Y\times X}: \text{ Assumptions \ref{as:Model}-\ref{as:px} hold with respect to } P \Big\}.
\end{align}
Then, by Assumption \ref{as:Model}, $(y^{n},x^{n})$ is generated from $P\in\mathcal P_\theta\subseteq\mathcal P_{Y\times X}$. Let $\mathcal P\equiv\{P:P\in\mathcal P_\theta,~\theta\in\Theta\}$ be the set of all distributions compatible with our restrictions.

\begin{thm}\label{thm:size_control}
Suppose Assumptions \ref{as:primitives}-\ref{as:px}
hold. Suppose $\tau_n>0$ for all $n$, $\tau_n\to \infty$ and $\tau_n=o(n^{1/2})$. Let $\alpha\in(0,1/2)$.
Let $\phi(y^n,x^n)=1\big\{\sup_{\theta\in\Theta}p_{n}(\theta)\le \alpha\big\}.$  Then, 
\begin{align}
\limsup_{n\to\infty}\sup_{P\in\mathcal P}E_{P}[\phi]\le\alpha.
\end{align}
\end{thm}

\begin{remark}\rm
Under $H_0$, the population criterion $T(\theta)$ in \eqref{eq:defT} equals zero for all $\theta\in\Theta$. We therefore construct a level-$\alpha$ test by testing $T(\theta)=0$ separately for each $\theta\in\Theta$ and rejecting $H_0$ only if all such tests reject. This procedure controls size by construction, though it may be conservative. When $H_0$ is not rejected, it yields a confidence set for $\theta$, and its computational cost can be mitigated by combining it with an appropriate optimization algorithm, as discussed in the next section.\footnote{An alternative is to use $\inf_{\theta\in\Theta} T_n(\theta)$ as a test statistic. While this approach may improve power, it requires resampling a suitable analog of $\inf_{\theta\in\Theta} T_n(\theta)$ and does not naturally yield a confidence region for $\theta$.}
\end{remark}

\subsection{Computational Aspects}\label{ssec:computation}
We discuss ways to simplify the computation of the statistic, implementation of the test, and construction of confidence intervals. Let
\begin{align}
    \vn& \equiv \sup_{b\in \balln}\sqrt n\{b^{\top}\hat P_{n,x}-h(b,\bce{S^r})\}\notag\\
    &=\sup_{b\in \balln}\inf_{q\in\bce{S^r}}\left[b^{\top}\hat P_{n,x}-b^{\top}q\right],  \quad(P0),\notag
\end{align}
and note that $T_n(\theta)=\sup_{x\in X}\vn.$
Following \cite{magnolfi2023estimation}, we may recast $(P0)$ as a convex quadratic program:
\begin{eqnarray}
\vn=\max_{lb\leq w\leq ub} & -\gamma^\top w & \label{eq:def_Vnx}\\
s.t. & w^\top \Gamma_1 w & \leq1\notag\\
 & \Gamma_2w & =0_{|Y|}\notag\\
 & \Gamma_3w & \le 0_{d_{\nu}},\notag
\end{eqnarray}
for some vector $\gamma$ and matrices $\Gamma_l$ for $l=1,2,3$; we provide details on how to construct these objects in Appendix \ref{app:comput_V}. The vector $w=(b^{\top},\lambda_{eq}^\top,\lambda_{ineq}^\top)^\top$ stacks $b\in\mathbb R^{|Y|-1}$ together with Lagrange multipliers associated with the constraints in the original problem. This reformulation simplifies the computation of the test statistic and, thus, the $p$-value function. 
Specifically, we solve the convex program for each $(\theta,x)$ and optimize the $p$-value function only once. This saves the number of times we need to solve a global non-convex optimization problem.\footnote{In contrast, directly applying a bootstrap procedure to the statistic $\inf_{\theta\in\Theta}T_n(\theta)$ is computationally demanding. It requires finding the global minimum of a non-linear (and typically non-convex) function $T_n(\theta)$ for a large number of times (essentially the number of bootstrap replications).}

Second, testing $H_0$ can be implemented by combining the $p$-value process $\theta \mapsto p_n(\theta)$ with a suitable optimization algorithm. For the purpose of testing $H_0$, it suffices to determine whether
\[
\sup_{\theta\in\Theta} p_n(\theta)
\]
exceeds the significance level $\alpha$.

When the parameter space $\Theta$ is low-dimensional, this supremum can be approximated by evaluating $p_n(\theta)$ over a sufficiently fine grid. More generally, when $\Theta$ has moderate dimension, one can rely on global optimization algorithms designed for black-box objective functions. Importantly, the optimization can be terminated as soon as a value of $\theta$ is found for which $p_n(\theta)>\alpha$, since this is sufficient to conclude that $H_0$ is not rejected.
While $p_n(\theta)$ does not admit a closed-form expression in general, a variety of global optimization methods are available for such settings. In particular, response surface methods are well suited for globally optimizing expensive black-box functions while keeping the number of function evaluations limited \citep{jones1998efficient}.

Finally, the response surface approach can also be used to compute confidence intervals—i.e., coordinate projections of $CS_n$—for individual components of $\theta$. When $\Theta$ is low-dimensional, such projections can be obtained using a grid over the parameter space. More generally, for moderate-dimensional $\Theta$, one can instead rely on a global optimization algorithm.
Specifically, the endpoints of a confidence interval for the $j$th component of $\theta$ can be obtained by solving
\begin{align*}
    \max/\min_{\theta\in\Theta} \;& e_j^{\top}\theta \\
    \text{s.t.} \;& p_n(\theta)\ge \alpha,
\end{align*}
where $e_j$ denotes the vector whose $j$-th component equals one and whose remaining components are zero. The constraint involves the black-box function $p_n(\theta)$, which generally lacks a closed-form expression. Response surface methods for solving such constrained optimization problems are developed in \cite{kaido2017calibrated}.

\section{Extensions}\label{sec:extensions}

\subsection{Testing Multiple Hypotheses with a Sequence of Baselines}\label{subsec:seq_baselines}

When the information structure is of independent economic interest,
the analyst may consider testing a sequence of hypotheses to refine
her understanding of the game's information structure. To motivate
this, let us revisit Example 3. Recall that $H_{0}:S\succeq S^{r}$
in this example is formed with $S^{r}=S_{P}$ to investigate whether
the hub airline player has informational advantage over the others. Yet,
one may argue that, even if $H_{0}$ is not rejected, it does not
definitively conclude that the hub player has informational privilege,
as the scenario where all the players participate in an entry game of
complete information remains plausible. This is because $S_{C}\succeq S_{P}$. Instead, suppose the analyst opts to
test two nulls of the form: $H_{0,j}:S\succeq S_{j}^{r}$ for $j=1,2$
where $S_{1}^{r}=S_{C}$ and $S_{2}^{r}=S_{P}$. If $H_{0,2}$ is
not rejected (as before) while $H_{0,1}$ is, then such test results
may offer stronger evidence towards the hub's privileged information.
To test each $H_{0,j}$, one can directly employ a $p$-value calculated
in Section \ref{ssec:test}. Given the test of multiple hypotheses, a remaining crucial task is to control the accumulated error of false rejections, namely,
the family-wise error rate (FWER).

Consider a sequence of null hypotheses, $H_{0,1},...,H_{0,J}$ where
\begin{align}
H_{0,j} & :S\succeq S_{j}^{r}.\label{eq:multi_hyp}
\end{align}
Motivated by the example above, we assume the following:
\begin{assumption}\label{as:info_ordering}
 $S_{1}^{r}\succeq\cdots\succeq S_{J}^{r}$.
\end{assumption}
Assumption \ref{as:info_ordering} establishes a logical relationship among the nulls; for
instance, if $H_{0,1}$ holds true, $H_{0,2}$ is automatically true.
As shown below, this assumption provides additional structure in multiple
testing, thereby improving the performance of the test. Instead of
sequentially testing $H_{0,1}$ and $H_{0,2}$, one may wish to test
a hypothesis of the form $S_{1}^{r}\succeq S\succeq S_{2}^{r}$. Although
this form of hypothesis cannot be accommodated within our framework,
we do not necessarily view this as a limitation of our framework.
Rather, we view this as reflecting the robustness property of the
framework, which can be seen, for instance, from Lemma \ref{lem:bce_bne}.

To control the FWER, we propose the following step-wise procedure.
This procedure uses a common cutoff to be compared to $p$-values
for all nulls but does not accumulate Type I errors by exploiting
logical dependence (Assumption \ref{as:info_ordering}). This leads to a test that is valid
but more powerful than standard \cite{holm1979simple}'s-type procedures that contol
FWER. We enjoy this benefit thanks to the ordering of the information structures.

Similar to Section \ref{ssec:test}, $H_{0,j}:S\succeq S_{J}^{r}$ only if $H_{0,j}:T_{n,j}(\theta)=0,~\forall\theta$
where $T_{n,j}$ is the test statistic analogous to \eqref{eq:Tn1} for $H_{0,j}$.
Let $p_{j}\equiv p_{n,j}\equiv\sup_{\theta\in\Theta}p_{n,j}(\theta)$
where $p_{n,j}(\theta)$ is the $p$-value analogous to \eqref{eq:p_value} that corresponds to $H_{0,j}$. Let $\alpha$ be the overall test's nominal level.

\begin{algorithm}[Level-$\alpha$ sequential test]\label{alg:seq_test}
Proceed as follows.

\noindent Step 1: Test $H_{0,1}$ at level $\alpha$. If $p_{1}>\alpha$, then
stop (and reject no nulls). Otherwise, reject $H_{0,1}$ and go to Step
2.

$\vdots$

\noindent Step $j$: Test $H_{0,j}$ at level $\alpha$. If $p_{j}>\alpha$, then
stop. Otherwise, reject $H_{0,j}$ and go to Step $j+1$.

$\vdots$

\noindent Step $J$: Test $H_{0,J}$ at level $\alpha$. If $p_{J}>\alpha$, then
stop. Otherwise, reject $H_{0,J}$ and stop.
\end{algorithm}

Under Assumption \ref{as:info_ordering}, there exists an index $j^{*}\in\{1,...,J\}$ such
that $H_{0,1},\dots,H_{0,j^{*}-1}$ are false nulls and $H_{0,j^{*}},...,H_{0,J}$
are true nulls. Then, according to Algorithm \ref{alg:seq_test}, $H_{0,j^{*}}$ is the
first true null that is being tested, in which case all the previous
nulls are correctly rejected. Therefore, there is no accumulation
of false rejection errors up to that step, and controlling the error
below $\alpha$ for $H_{0,j^{*}}$ is sufficient to control for the
overall false rejection rate. Formally,
\begin{equation}\label{eq:FWER}
FWER\equiv P\left(\text{reject at least one true }H_{0,j}\right)=P\left(\text{reject }H_{0,j^{*}}\right)\le\alpha.    
\end{equation}
The last inequality holds because a level-$\alpha$ test is conducted
in each step and by Theorem \ref{thm:size_control}. The second equality holds by the following
argument. First, to reject any true null, we must first reach a true
null with the smallest $j$, namely, $H_{0,j^{*}}$. Reaching any $j>j^{*}$
implies that $H_{0,j^{*}}$ has already been rejected. Conversely, if
you do not reject $H_{0,j^{*}}$, you stop immediately, and no other
true null is ever tested. Hence rejecting at least one true null is
equivalent to rejecting $H_{0,j^{*}}$. Note that we do not need correction
as in \cite{holm1979simple}; testing each $H_{0,j}$ at level $\alpha$ and stopping
on the first non-rejection suffices. This aspect is appealing especially when one considers a fine grid of benchmark information structures (i.e., large $J$). A similar
approach appears in \cite{goeman2010sequential} and \cite{hansen2011model}. Let $\phi^{MH}(y^{n},x^{n})=1\{\text{Algorithm \ref{alg:seq_test} rejects some nulls}\}$.
The following theorem shows that the proposed test controls size.

\begin{thm}\label{thm:size_control_multi}
 Suppose Assumptions 1-\ref{as:info_ordering} hold. The modified Holm's procedure
satisfies
\begin{align*}
\limsup_{n\rightarrow\infty}\sup_{P\in\mathcal{P}}E_{P}[\phi^{MH}] & \le\alpha.
\end{align*}
\end{thm}
Although the theory allows for a general $J$, we note that $J=2$
case can already yield more convincing and interpretable conclusions
from test outcomes compared to the $J=1$ case, as the example above
suggests.
\begin{remark} Continuing the example above, one may wish to alternatively
set $S_{1}^{r}=S_{P_{1}}$ and $S_{2}^{r}=S_{P_{2}}$. Again, if $H_{0,1}$
is rejected but  $H_{0,2}$ is not, we can at least rule out the possibility
of not rejecting $H_{0,2}$ because of complete information of all
players. However, these nulls do not satisfy Assumption 6 and thus
the test cannot exploit any logical relationship. We can still apply
the original Holm's procedure, modified for the specification test
with a partially identified model.
\end{remark}

\subsection{Confidence Sets for Markets with a Specified Baseline}\label{subsec:CS_markets}

So far, we have considered testing if $S(x)\is S^r(x)$ across all $x\in X$ (or a known subset of $X$), which leads to a binary decision. In some applications, one may want to explore the following question: Given a baseline information structure $S^r(x),x\in X$, what is the set of markets where the true information structure is at least as informative as $S^r(x)$?
Specifically, for each $x\in X$, consider the hypothesis
\begin{align}
H_{0,x}: S(x)\succeq S^r(x).
\end{align}
Let $X_0$ be the set of $x$'s for which $H_{0,x}$ is true. Our objective here is to construct a confidence region for $X_0$.

Letting $T_x(\theta)\equiv\sup_{b\in \ball}\{b^{\top}P_{y|x}-h(b,Q^{BCE}_{\theta,S^r}(x))\}$ be the one-sided population measure of the distance between the CCP and BCE prediction, we can recast the problem as testing a family of hypotheses $H_{0,x}: T_x(\theta)=0$ across $x\in X$.
Let 
\begin{align}
	T_{n,x}(\theta)&=\sup_{b\in \balln}\sqrt n\{b^{\top}\hat P_{n,x}-h(b,\bce{S^r})\}\\
	T^*_{n,x}(\theta)&=\sup_{b\in \Psi_{n,x,\theta}}\{\mathbb G^*_n(b,x)\}
\end{align}
 be the sample and bootstrap analogs of $T_x(\theta)$, where $\Psi_{n,x,\theta}\equiv\{b:\hat \eta_{n,\theta}(b,x)\ge -\tau_n,b\in\balln\}.$ These quantities resemble $T_n$ and $T^*_n$ but do not aggregate the statistics across $X$.
  One can then define tests for individual hypotheses and combine them to construct a confidence region. Specifically, define the p-value: $p_n(\theta,x)=P^*(T^*_{n,x}(\theta)>T_{n,x}(\theta)|y^n,x^n)$. Then, $p_n(x)=\sup_{\theta\in\Theta}p_n(\theta,x)\le \alpha$ is an asymptototically valid level-$\alpha$ test for each $H_{0,x}$.
 
We define a confidence set as follows:
\begin{align}
	\text{CS}_n\equiv\big\{x\in X:p_n(x)> \alpha_x\big\},
\end{align}
where $\alpha_x$ is chosen to control the FWER accounting for multiple comparisons.
For example, common choices are the Bonferroni ($\alpha_x=\frac{\alpha}{|X|}$) and Holm ($\alpha_{x_{(k)}}=\frac{\alpha}{|X|-k+1}$) procedures.
For the Holm procedure, the $x$-values are ordered according to the corresponding p-values: $p_n(x_{(1)})\le p_n(x_{(2)})\le \dots\le p_n(x_{(|X|)})$. This procedure starts from the smallest p-value, rejects $H_{0,x}$ if $p_n(x_{(k)})\le \alpha_{x_{(k)}}$, and  stops at the step in which one does not reject for the ﬁrst time. We then keep all remaining $x$ values in $\text{CS}_n$. The standard argument ensures that these confidence sets have the following coverage property:
\begin{align}
	\liminf_{n\to\infty}\inf_{P\in \mathcal P}P(X_0\in \text{CS}_n)\ge 1-\alpha.
\end{align}

\section{Monte Carlo Experiments}\label{sec:simulation}

\subsection{Simulation Design}
To illustrate our method, we construct a class of information structures where players receive a private signal in addition to knowing their own payoff types. In the limit, the private signal is perfectly informative about a portion of the opponent's payoff type, making it public information as in $S_{Pub}$ in Example 2. This limit information structure is adopted by \cite{aguirregabiria2007sequential} in the context of dynamic entry games.

To fix ideas, consider the two-player entry game of Example~1. Firm $i$ earns zero profits when not entering, and earns
$\pi_i(y,\varepsilon_i;x,\theta_\pi)
=
x\beta+\Delta y_{-i}+\varepsilon_i$
upon entry, where $x$ is a random variable supported on $\{-M,M\}$. Suppose that the payoff type $\varepsilon_i$ admits the decomposition $\varepsilon_i=\nu_i+\epsilon_i$  where $\nu_i$ is a discrete component for which a player can receive the opponent's signal. Let $\nu_i$ take values in $\{-\eta,\eta\}$ for some $\eta>0$, corresponding to low and high states, and suppose
\[
P(\nu_i=\eta)=\mu\in[0,1].
\]
The component $\nu_i$ is i.i.d.\ across players. The opponent receives a signal $t_i$ about this component.
The remaining component $\epsilon_i$ is a standard normal random variable, also independent across $i$. The structure of the game is common knowledge among the players.

Consider an information structure in which, in addition to knowing their own payoff type, each player observes an informative signal $t_i$ with support $\{\eta,-\eta\}$ that conveys information about the realization of $\nu_{-i}$. In particular, the signal reveals $\nu_{-i}$ with probability $\xi$, that is,
\[
P(t_i=\nu \mid \nu_{-i}=\nu)=\xi, \qquad \nu\in\{\eta,-\eta\}.
\]
Thus, upon observing $t_i=\eta$, player $i$ forms the posterior belief
\begin{align*}
P(\nu_{-i}=\eta \mid t_i=\eta)
&=
\frac{\mu \xi}{\mu \xi+(1-\mu)(1-\xi)}
\;\equiv\;
\rho_\eta(\mu,\xi)\\
P(\nu_{-i}=-\eta \mid t_i=-\eta)
&=
\frac{(1-\mu)\xi}{(1-\mu)\xi+\mu(1-\xi)}
\;\equiv\;
\rho_{-\eta}(\mu,\xi).
\end{align*}
We refer to this information structure as $S_\xi$. Note that as $\xi\to \tfrac12$, $\rho_\eta(\mu,\xi)\to \mu$. Hence, as $\xi$ approaches $1/2$, $S_\xi$ converges to the incomplete information structure $S_I$, under which each player only knows their own type.

For this information structure, Bayesian Nash equilibria are characterized by threshold strategies of the form
\begin{align}
y_i
=
\mathbbm 1\{\varepsilon_i \ge \tau_i(\nu_i,t_i)\},
\qquad i=1,2,
\label{eq:BNE_threshold}
\end{align}
where $\tau_i(\nu_i,t_i)$ is defined in Appendix~\ref{sec:Appendix_MC}. Let $\sv(\tau)\equiv P(\varepsilon_i\ge\tau)$.
The equilibrium probability of entry for a player conditional on the payoff state $\nu_i$ is given by
\begin{align*}
P_{\nu_i=\eta,x}
&=
\xi\,\sv(\tau(\eta,\eta))+(1-\xi)\,\sv(\tau(\eta,-\eta)),\\
P_{\nu_i=-\eta,x}
&=
\xi\,\sv(\tau(-\eta,-\eta))+(1-\xi)\,\sv(\tau(-\eta,\eta)).
\end{align*}

We simulate a Bayesian Nash equilibrium outcome using the conditional distribution
\begin{align}
P_{y\mid x}(y\mid x)
=
\int
\prod_{i=1}^2
\bigl(P_{\nu_i=\eta,x}\bigr)^{y_i}
\bigl(1-P_{\nu_i=\eta,x}\bigr)^{1-y_i}
\, dP(\nu),
\qquad
y=(y_1,y_2)^\top\in Y .
\end{align}
This probability depends on the informativeness of the signal, measured by $\xi\in[0.5,1]$. Figure~\ref{fig:ccp_and_bce} (top panel) illustrates how the equilibrium CCP varies with $\xi$. As signal quality increases, the entire CCP vector moves from the CCPs associated with low-quality signals (blue dots) toward those associated with high-quality signals (yellow dots).

We test the null hypothesis that $S \succeq S_{Pub}$, that is, that the true information structure $S$ is at least as informative as the public information structure $S_{Pub}$. Under $S_{Pub}$, the payoff components $\nu_i$ are publicly observed. This benchmark corresponds to the case $\xi=1$.

Figure~\ref{fig:ccp_and_bce} also plots the BCE prediction for a fixed value of $\theta$. As expected, the BCE prediction contains the CCP corresponding to $\xi=1$. It also contains some CCPs generated under relatively informative private signals about $\nu_{-i}$. Consequently, no test can be expected to have power against alternatives that are close to the null in terms of signal quality.

By contrast, meaningful power can be attained when the equilibrium CCP lies outside the BCE prediction under $S^{r}$, which occurs when $\xi$ departs sufficiently from $1$ toward $0.5$. Intuitively, as signals become less informative, players rely more heavily on their priors about opponents’ types, leading to systematic shifts in entry behavior that are no longer compatible with the BCE prediction under the public-information benchmark.

Below, we examine the size and power of the test when the data are generated from a Bayesian Nash equilibrium with $\xi$ varying over $[0.5,1]$.

\begin{figure}[htbp]
\centering
\caption{Equilibrium CCPs and the BCE Prediction}
\label{fig:ccp_and_bce}
\vspace{0.5em}
\begin{tabular}{cc}
\textbf{(a)} Full simplex view & \textbf{(b)} Zoomed view \\[0.5em]
\includegraphics[width=0.48\textwidth]{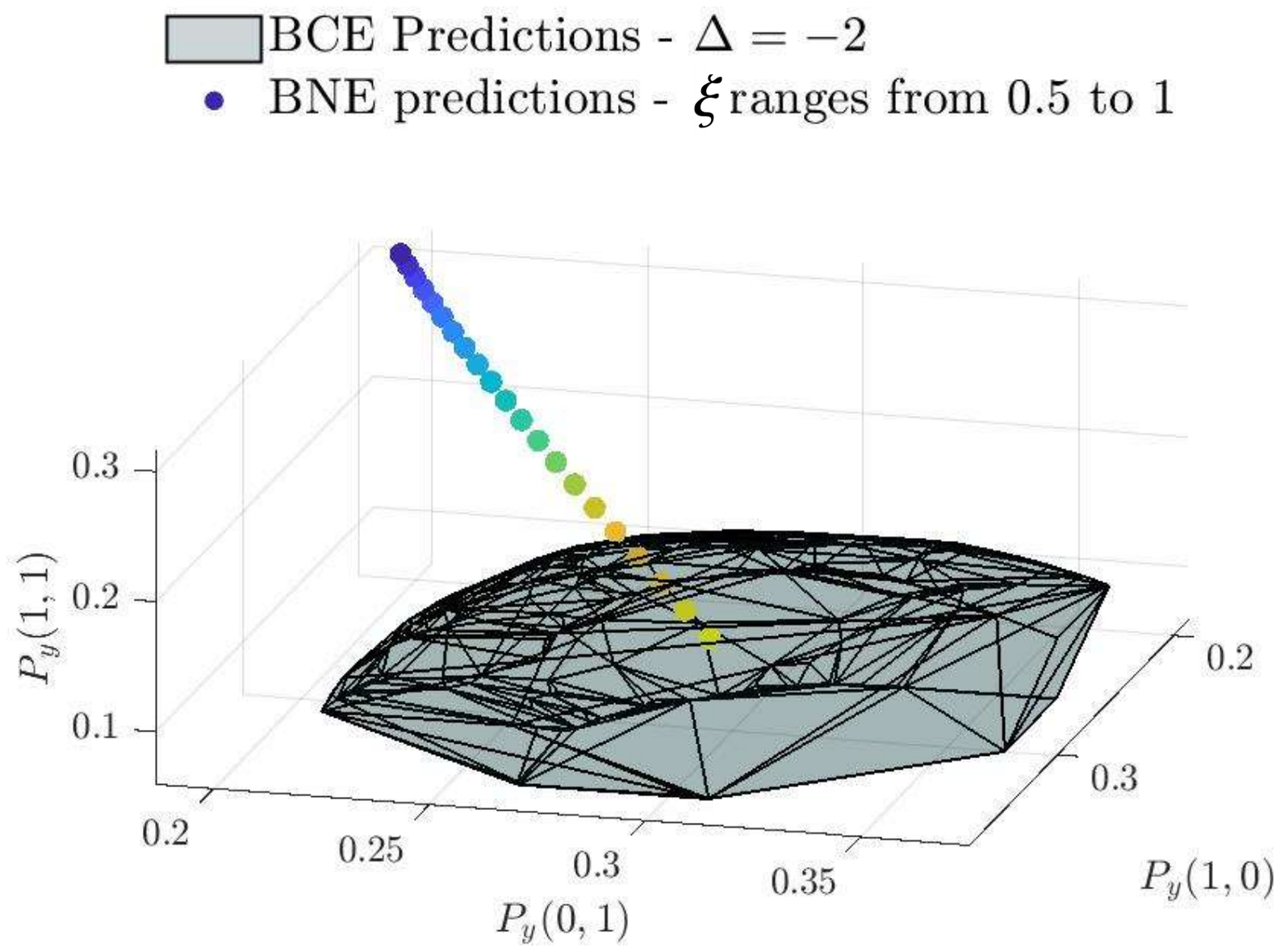} &
\includegraphics[width=0.48\textwidth]{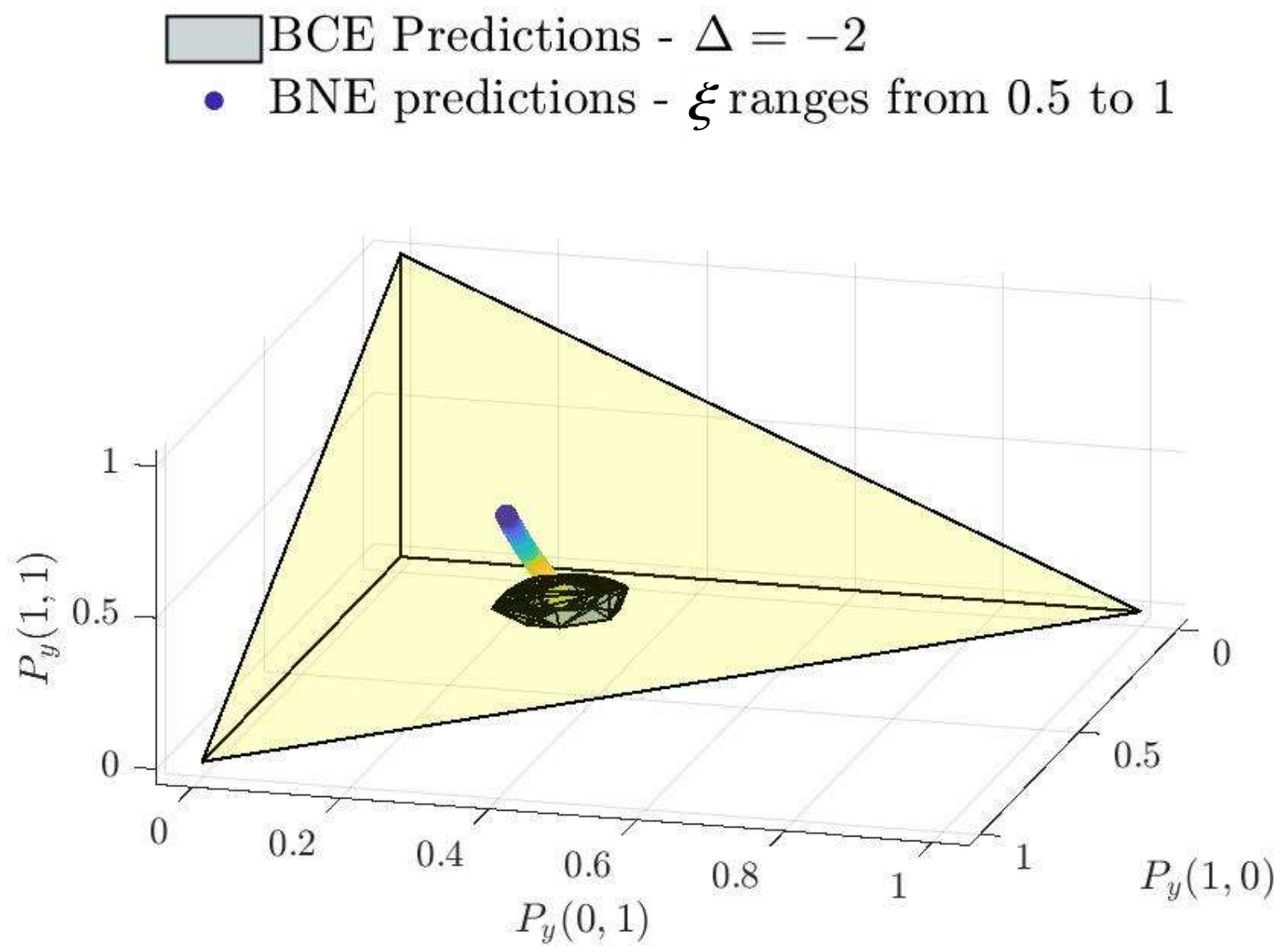}
\end{tabular}

\vspace{0.5em}
\begin{minipage}{0.95\textwidth}
\small
\textit{Notes:} The shaded region represents the BCE prediction set for $\Delta = -2$. Dots represent BNE conditional choice probabilities as signal quality $\xi$ varies from 0.5 (blue) to 1 (yellow). Panel (a) shows the BCE prediction within the probability simplex; panel (b) provides a zoomed view near the boundary. As $\xi$ increases toward 1, the CCP moves from outside the BCE prediction into the set.
\end{minipage}
\end{figure}

\subsection{Simulation Results}

Figure \ref{fig:power_M3} reports the power of the proposed test for a sample of size of 1,000 with $M=1,2,$ and 3. The null hypothesis corresponds to $\xi=1$, and the significance level is set to $0.05$. The test controls size uniformly across all values of $M$, although it is conservative at $\xi=1$, as expected.

Recall from Figure~\ref{fig:ccp_and_bce} that the conditional choice probabilities remain inside the BCE prediction for large values of $\xi$. 
Consistent with this observation, the rejection probability begins to increase only once $\xi$ becomes sufficiently small. It then rises rapidly as $\xi$ decreases further, indicating that the test is able to detect local deviations from the boundary of the BCE prediction.

Figure~\ref{fig:power_M3} also reveals a nonmonotonic interaction between signal informativeness $\xi$ and the support of the covariates. When $\xi$ is either sufficiently high or sufficiently low (specifically, $\xi$ above $0.9$ or below $0.75$), a larger support ($M=3$) yields slightly higher rejection probabilities than the other designs. In contrast, for intermediate values of $\xi$—roughly between $0.775$ and $0.875$—the test achieves higher power when $M=1$. Interestingly, for almost all values of $\xi$, the test exhibits the lowest power when $M=2$.

\begin{figure}[htbp]
\centering
\caption{Rejection Probability of the Test}
\label{fig:power_M3}
\vspace{0.5em}
\includegraphics[trim={0 3.83in 0 3.85in},clip,width=\textwidth]{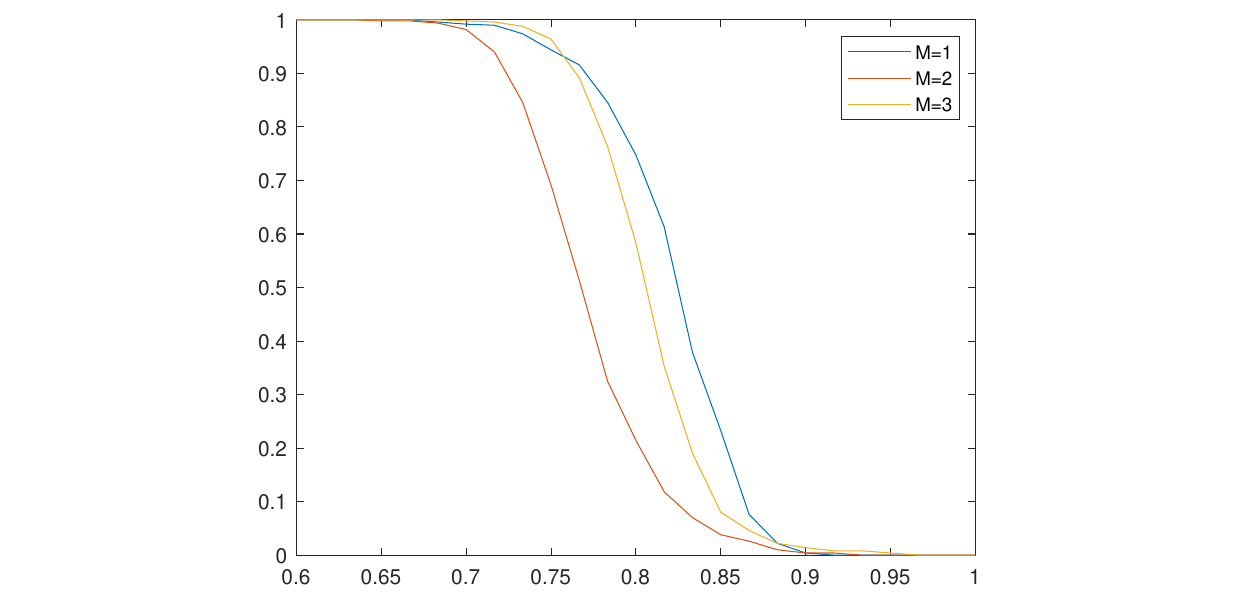}

\vspace{0.5em}
\begin{minipage}{0.95\textwidth}
\small
\textit{Notes:} The horizontal axis is signal quality $\xi$, ranging from 0.5 (incomplete information) to 1 (public information, the null hypothesis). The vertical axis is the rejection probability. Sample size is $n = 1{,}000$ and the nominal level is 5\%. The three lines correspond to different supports for the covariate $x \in \{-M, M\}$.
\end{minipage}
\end{figure}
\section{Application: Information in Airline Markets}\label{sec:application}

\subsection{Economic Environment and Research Question}

The US airline industry provides a compelling setting to examine information asymmetries in strategic entry decisions. Following deregulation in the 1970s, the industry rapidly evolved toward a hub-and-spoke network structure, where major airlines concentrate operations at specific airports. This market structure raises an important question: do hub airlines possess informational advantages beyond their well-documented cost and demand benefits?

The economics literature has established that hub operations confer significant advantages. \cite{Berry:1996aa} show that hub airlines benefit from both demand-side advantages (through improved service quality for time-sensitive travelers) and cost-side economies. These advantages stem from increased flight frequency, more convenient connections, and operational efficiencies at hub airports. However, whether hub presence also provides superior information about market conditions remains an open question.

Entry decisions in airline markets are inherently strategic, as emphasized in the seminal work of \cite{berry1992estimation} and further developed by \cite{ciliberto2009market} and \cite{ciliberto2021market}. An airline's profitability depends crucially on which competitors choose to serve the same route. The discrete nature of the entry decision and the observable market structure align well with our theoretical framework for testing information ordering. This environment thus represents an ideal application for our method.  

Our analysis abstracts from the full network structure of airline operations to focus specifically on information asymmetries. While network effects induce important correlations in firms' decisions across markets—as modeled in  \cite{bontemps2023online} and \cite{yuan2024network}—we treat each market independently to isolate the role of information in entry decisions.

\subsection{Data and Descriptive Evidence}\label{ssec:data}

We analyze airline entry decisions using data from the Department of Transportation's Origin and Destination Survey (DB1B) and Domestic Segment (T-100) database for the first quarter of 2019. Following the empirical literature on airline competition \citep[e.g.,][]{berry1992estimation,ciliberto2009market}, we define a market as a directional city-pair route and consider an airline as actively serving a market if it carries at least 90 passengers in the quarter.

To maintain a parsimonious specification while capturing the key economic forces, we aggregate airlines into three strategic players in each market. The \emph{Hub Major} player represents a major airline (American, Delta, or United) for which at least one endpoint of the route is a hub. The \emph{Non-Hub Major} player aggregates major airlines without hub presence in the market. The \emph{Non-Major} player comprises low-cost carriers, JetBlue, Alaska and Southwest Airlines. Our analysis focuses exclusively on hub markets---routes where at least one endpoint serves as a hub for a major airline---yielding 1,810 market observations. This restriction ensures that every market features a potential hub player, allowing us to test hypotheses about informational advantages conferred by hub status.

For observable covariates, we follow \cite{berry1992estimation} and \cite{ciliberto2009market} and include market size, measured as the geometric mean of endpoint MSA populations, as a market-level characteristic. For player-specific variables, we use airport presence, defined as the fraction of markets an airline serves from an airport. All continuous variables are discretized at their medians to create binary indicators, yielding $2^4 = 16$ distinct market types based on the following characteristics: market size, hub airline presence, non-hub major presence, and non-major presence.

Entry rates differ markedly across player types. Non-Major carriers exhibit the highest entry rate at 79.8\%, reflecting the prevalence of low-cost carrier service in hub markets. Hub Major airlines enter 42.5\% of markets where they have hub presence, while Non-Hub Major airlines enter only 15.2\% of markets, consistent with the difficulty of competing against an established hub carrier.

Table \ref{tab:market_structure_new} displays the distribution of market structures across all markets and by key covariates. The most common market structure involves only the Non-Major player operating (47.0\% of markets), followed by duopoly competition between the Hub Major and Non-Major (23.3\%). Markets with no entry account for only 3.9\% of observations, while markets with all three player types present account for 4.9\%.

The distribution of market structures varies substantially with observable characteristics. In small markets, Non-Major monopoly is particularly prevalent (60.8\%), while larger markets exhibit more competitive structures with higher rates of multi-player entry. Hub airline presence strongly predicts market structure: in markets where the hub airline has high presence, Hub Major entry (alone or with others) occurs in 72.3\% of cases, compared to only 13.0\% when hub presence is low. Conversely, Non-Major monopoly drops from 69.0\% in low hub-presence markets to 24.8\% in high hub-presence markets.

\begin{table}[htbp]
\begin{threeparttable}
\caption{Distribution of Market Structures}
\label{tab:market_structure_new}
\centering
\small
\begin{tabular}{lccccccccc}
\hline\hline
 & \multicolumn{8}{c}{Market Structure (Hub, Non-Hub Major, Non-Major)} & \\
\cline{2-9}
 & (0,0,0) & (0,0,1) & (0,1,0) & (1,0,0) & (0,1,1) & (1,0,1) & (1,1,0) & (1,1,1) & $N$ \\
\hline
All Markets & 3.9 & 47.0 & 2.0 & 10.6 & 4.6 & 23.3 & 3.7 & 4.9 & 1,810 \\[1ex]
\multicolumn{10}{l}{\textit{By Market Size}} \\
\quad Small & 6.9 & 60.8 & 0.6 & 8.6 & 2.0 & 19.3 & 0.8 & 1.1 & 905 \\
\quad Large & 0.9 & 33.3 & 3.5 & 12.6 & 7.2 & 27.2 & 6.6 & 8.7 & 905 \\[1ex]
\multicolumn{10}{l}{\textit{By Hub Major Presence}} \\
\quad Low & 6.8 & 69.0 & 2.9 & 0.7 & 8.4 & 6.4 & 1.4 & 4.5 & 910 \\
\quad High & 0.9 & 24.8 & 1.2 & 20.7 & 0.8 & 40.3 & 6.0 & 5.3 & 900 \\
\hline\hline
\end{tabular}
\begin{tablenotes}[flushleft]
\small
\item \textit{Notes:} Sample includes 1,810 hub markets from DOT DB1B and T-100 for Q1 2019. Entries are percentages. Market structure is denoted (Hub, Non-Hub Major, Non-Major), where 1 indicates entry and 0 indicates no entry. For example, (1,0,1) indicates that the Hub Major and Non-Major players entered while the Non-Hub Major did not. Market size and hub presence are binary indicators split at the median.
\end{tablenotes}
\end{threeparttable}
\end{table}

Several patterns in the descriptive evidence merit discussion. First, Non-Major carriers are nearly ubiquitous, entering nearly 80\% of hub markets. This reflects the competitive pressure that low-cost carriers exert even in markets dominated by hub airlines. Second, the strong correlation between hub presence and hub entry suggests that hub airlines' entry decisions are closely tied to their existing network structure. Third, the relatively low rate of Non-Hub Major entry (15.2\%) indicates that competing against an established hub airline is challenging, consistent with the cost and demand advantages documented in the literature.

These patterns motivate our test of information ordering. If hub airlines possess informational advantages beyond their cost and demand benefits, we would expect this to manifest in systematic differences in entry behavior that cannot be rationalized under the baseline information structure. The variation in market structures across the 16 market types provides the identifying variation for our test.

\subsection{Testing Information Ordering}\label{ssec:test_application}

We now apply the testing framework developed in Section \ref{ssec:test} to examine whether hub airlines have informational advantages in entry decisions. The null hypothesis posits that the true information structure $S$ is at least as informative as a baseline structure $S^r$ featuring privileged information for hub airlines:
\begin{align}
H_0: S \succeq S^r \quad \text{versus} \quad H_1: S \nsucceq S^r,
\end{align}
where $S^r = S_P$ denotes the privileged information structure defined in Section 2.1. Under $S_P$, the hub airline observes the complete vector of payoff shocks $\varepsilon = (\varepsilon_1, \varepsilon_2, \varepsilon_3)$, while non-hub players observe only their own payoff realizations. This specification captures the hypothesis that hub presence confers informational advantages---beyond the cost and demand benefits documented in the literature---by providing better knowledge of competitors' profitability.

Under the null hypothesis, the observed conditional choice probabilities must satisfy $P_{y|x} \in \bce{S_P}$ for some parameter vector $\theta \in \Theta$ and all market types $x \in X$. As established in Section \ref{ssec:test}, we construct the test statistic $T_n(\theta)$ defined in \eqref{eq:Tn1}, which measures the maximum deviation of the empirical conditional choice probabilities from the BCE prediction set. The null hypothesis is rejected when the bootstrap $p$-value $p_n(\theta) \leq \alpha$ for all $\theta \in \Theta$.

Our empirical setting features 16 distinct market types defined by the binary covariates described in Section \ref{ssec:data}. We implement two complementary approaches. First, we test the joint null hypothesis that $S(x) \succeq S_P(x)$ for all market types simultaneously. Second, following the framework developed for confidence sets in Section \ref{subsec:CS_markets}, we construct confidence sets for the markets where the privileged information assumption holds, applying both the Bonferroni and Holm procedures to control the family-wise error rate at level $\alpha = 0.05$.

\subsection{Test Results}

Table \ref{tab:test_results_new} reports the test results for each of the 16 market types. The columns display market characteristics, sample size, the $p$-value from the market-specific test, and the decisions under both the Bonferroni and Holm procedures.

\begin{table}[htbp]
\begin{threeparttable}
\caption{Test Results for Privileged Information Structure}
\label{tab:test_results_new}
\centering
\small
\begin{tabular}{cccccrccc}
\hline\hline
 & \multicolumn{4}{c}{Market Characteristics} & & & \multicolumn{2}{c}{Decision} \\
\cline{2-5} \cline{8-9}
Market & Size & Hub Pres. & NHM Pres. & NM Pres. & $n$ & $p$-value & Bonferroni & Holm \\
\hline
1  & Small & Low  & Low  & Low  &  82 & 0.602 & Accept & Accept \\
2  & Small & Low  & Low  & High & 233 & 0.177 & Accept & Accept \\
3  & Small & Low  & High & Low  &  66 & 0.054 & Accept & Accept \\
4  & Small & Low  & High & High & 128 & 0.495 & Accept & Accept \\
5  & Small & High & Low  & Low  & 105 & 0.019 & Accept & Accept \\
6  & Small & High & Low  & High & 144 & 0.018 & Accept & Accept \\
7  & Small & High & High & Low  &  88 & 0.347 & Accept & Accept \\
8  & Small & High & High & High &  59 & $<$0.001 & Reject & Reject \\[0.5ex]
9  & Large & Low  & Low  & Low  &  54 & 0.516 & Accept & Accept \\
10 & Large & Low  & Low  & High &  95 & 0.017 & Accept & Accept \\
11 & Large & Low  & High & Low  & 161 & $<$0.001 & Reject & Reject \\
12 & Large & Low  & High & High &  91 & $<$0.001 & Reject & Reject \\
13 & Large & High & Low  & Low  &  98 & $<$0.001 & Reject & Reject \\
14 & Large & High & Low  & High &  94 & $<$0.001 & Reject & Reject \\
15 & Large & High & High & Low  & 252 & $<$0.001 & Reject & Reject \\
16 & Large & High & High & High &  60 & $<$0.001 & Reject & Reject \\
\hline\hline
\end{tabular}
\begin{tablenotes}[flushleft]
\small
\item \textit{Notes:} The null hypothesis is $H_0: S(x) \succeq S_P(x)$, where $S_P$ denotes the privileged information structure under which the hub airline observes all payoff shocks while other players observe only their own. The test is conducted at the 5\% significance level. The $p$-value column reports $\sup_{\theta \in \Theta} p_n(\theta, x)$ for each market type $x$. For $p$-values reported as $<$0.001, the numerical value was below the precision threshold. For $p$-values exceeding 0.05, the reported values are lower bounds, as the optimization terminates once $p_n(\theta) > \alpha$. The Bonferroni procedure uses critical value $\alpha/16 = 0.003125$; the Holm procedure applies the sequential algorithm described in Section \ref{subsec:CS_markets}.
\end{tablenotes}
\end{threeparttable}
\end{table}

The results provide strong evidence against the privileged information hypothesis. Under both Bonferroni and Holm precedures, we reject the null hypothesis in seven markets: 8 and 11--16. 
The pattern of rejections reveals systematic heterogeneity across market types. All rejections occur in markets characterized by either large size or high hub airline presence, and typically both. Specifically, among large markets, six out of eight reject under the Holm procedure---the exceptions being markets 9 and 10, which feature low hub presence and low non-hub major presence. The only small-market rejection (market 8) occurs in the only small market where all players have high presence. These patterns suggest that the privileged information assumption is most strongly violated in competitive environments where multiple players are well-positioned to serve the market.

To interpret these findings, recall that rejection of $H_0: S(x) \succeq S_P(x)$ implies that the observed conditional choice probabilities lie outside the BCE prediction set $\bce{S_P}$. This can occur for two distinct reasons. First, the true information structure may be \emph{less} informative than $S_P$---for instance, if hub airlines do not actually observe competitors' payoff shocks. Second, the true information structure may be incomparable to $S_P$ under the partial order $\succeq$, which would arise if information is distributed differently than the privileged structure assumes.

The concentration of rejections in large, competitive markets points toward the first interpretation. In these markets, we observe entry patterns (particularly the frequency of simultaneous entry by hub and non-major carriers) that are difficult to reconcile with the hub airline having complete information about competitors' profitability. If hub airlines knew that entry by non-major carriers would be profitable, the strategic response under complete information would perhaps differ from the observed patterns. This suggests that hub airlines' entry decisions are based on less precise information about competitors than the privileged structure assumes, consistent with an information structure closer to incomplete information where each player observes primarily their own payoff shock.

These results have implications for counterfactual policy analysis. Researchers seeking to simulate entry behavior under alternative market conditions should be cautious about assuming that hub airlines have superior information. Our findings suggest that a more symmetric information structure, where hub status confers cost and demand advantages but not substantial informational advantages, may better approximate the true strategic environment in airline markets.

\section{Conclusion}\label{sec:conclusion}
The information available to players is a key primitive of any strategic environment. While it is difficult to specify the information structure precisely, data can provide guidance on whether the true information structure exceeds a proposed baseline. This paper develops a formal testing framework for information ordering in discrete games, leveraging the robust prediction properties of Bayes Correlated Equilibrium to translate hypotheses about information structures into testable moment inequalities. Our bootstrap-based testing procedure achieves uniform asymptotic validity across a broad class of data-generating processes, and we extend the framework to accommodate sequential testing of multiple hypotheses and the construction of confidence sets for markets satisfying a given baseline. Monte Carlo experiments demonstrate that the test properly controls size and exhibits meaningful power against local alternatives.

Our application to the U.S. airline industry illustrates the practical value of the methodology. Testing whether hub airlines have informational advantages beyond their well-documented cost and demand benefits, we find strong evidence against the privileged information hypothesis, with rejections concentrated in large, competitive markets. These results suggest that researchers conducting counterfactual policy analysis in airline markets should be cautious about assuming that hub status confers superior information about competitors' profitability. More broadly, our framework provides a tool for detecting asymmetric information in a range of strategic settings without imposing restrictive parametric assumptions on the underlying information structure.

\bibliography{bceref,bceref2}

\clearpage
\appendix

\section{Computation}

\subsection{Computing $V_{n,x}(\theta)$}\label{app:comput_V} 

The object $V_{n,x}(\theta)$ is the value of the following
maxmin program 
\begin{align*}
V_{n,x}(\theta)=\max_{b\in \balln}\min_{q\in \bce{S^r}}\left[b^{\top}\hat P_{n,x}-b^\top q\right], & \quad(P0)
\end{align*}
which can be computed for every value of $x$ and $\theta$.

\textbf{Step 1 - Discretization:} To make $(P0)$ feasible we approximate
the infinite dimensional object $\nu$ by discretizing the set $\mathcal{E}$.
Let $\mathcal{E}^{r}\subset\mathcal{E}$ be the discretized set, with
$|\mathcal{E}^{r}|=r$. We obtain $\mathcal{E}^{r}$ as the product
space of $\mathcal{E}_{i}^{r},$ where every set $\mathcal{E}_{i}^{r}$
contains $r_{i}=\frac{r}{|N|}$ equally spaced quantiles of $F_{\varepsilon_{i}}$.\footnote{We have experimented with other discretization techniques (e.g. Halton
sets, random draws) and have found negligible impact on our results
as long as $\mathcal{E}^{r}$ includes at least some relatively extreme
(both postive and negative) payoff types. This is because the incentive
compatibility constraint of BCE is more likely to be binding for these
values.} We also define $f^{r}\left(\cdot;\theta_{\varepsilon}\right)$ as
the probability mass function over $\mathcal{E}^{r},$ where the mass
of each $\varepsilon\in\mathcal{E}^{r}$ is generated by a Normal
copula with correlation parameter $\rho=\theta_{\varepsilon}.$ The
program $(P0)$ is then approximated by the feasible program 
\begin{eqnarray*}
\max_{b\in\mathbb{R}^{|Y|}}\min_{q\in\mathbb{R}^{|Y|},\nu\in\mathbb{R}^{|Y|\times r}} & b^\top\left(\hat P_{n,x}-q\right) & (P1)\\
s.t. & b^\top W_{n,x}b-1 & \leq0\\
\forall y\in Y & q\left(y\right)-\sum_{\varepsilon}\nu\left(y,\varepsilon\right) & =0\\
\forall\varepsilon\in\mathcal{E}^{r} & \sum_{y}\nu\left(y,\varepsilon\right)-f^{r}\left(\varepsilon;\theta_{\varepsilon}\right) & =0\\
 & \sum_{y,\varepsilon}\nu\left(y,\varepsilon\right)-1 & =0\\
\forall i,y_{i},y_{i}',\varepsilon_{i} & \sum_{y_{-i}}\sum_{\varepsilon_{-i}}\nu\left(y,\varepsilon_{i},\varepsilon_{-i}\right)\left(\pi_{i}\left(y_{i}^{\prime},y_{-i},\varepsilon_{i};x,\theta\right)-\pi_{i}\left(y,\varepsilon_{i};x,\theta\right)\right) & \leq0.
\end{eqnarray*}
Although in $(P0)$ the minimum is taken over $q\in Q_{\theta}^{BCE}\left(x\right)$
only, here we minimize over both a vector of predictions $q\in\mathcal{P}_{Y}$
and a distribution $\nu\in\mathcal{P}_{Y\times\mathcal{E}^{r}}$ whose
marginal on $Y$ corresponds to $q.$ The restriction that $q$ must
be a BCE prediction is now incorporated by imposing that $\nu$ must
satisfy the constraints that characterize BCE distributions, as specified
in Definition 2.

\textbf{Step 2 - Vectorization:} The discretized $\nu$ is a matrix
with dimensions $|Y|\times r$; we define $v=\mbox{vec}\left(\nu\right),$
the vectorized $\nu$ that stacks the columns of $\nu$ in a vector
with $d_{\nu}=|Y|\cdot r$ rows. We further specify how this vector
is constructed; the vector $v$ is a column vector formed by $r$
columns 
\[
\nu_{y}\left(\varepsilon\right)=\left[\begin{array}{c}
\nu\left(y^{1},\varepsilon\right)\\
\nu\left(y^{2},\varepsilon\right)\\
..\\
..\\
\nu\left(y^{|Y|},\varepsilon\right)
\end{array}\right]
\]
where vectors of actions are orderer in an order $y^{1}...y^{|Y|},$
and for a specific value of $\varepsilon$. The full $v$ is then:
\[
v=\left[\begin{array}{c}
\nu_{y}\left(\varepsilon^{1}\right)\\
\nu_{y}\left(\varepsilon^{2}\right)\\
..\\
..\\
\nu_{y}\left(\varepsilon^{r}\right)
\end{array}\right].
\]
To define orderings for both $\varepsilon$ and $y,$ we start from
complete orderings of both $\mathcal{E}_{i}$ and $Y_{i}$, summarized
by the respective sets of indices. Then, we order the vectors as follows:
\begin{align*}
\varepsilon^{1}=\left(\varepsilon_{1}^{1},...,\varepsilon_{|N|}^{1}\right); & y^{1}=\left(y_{1}^{1},...,y_{|N|}^{1}\right)\\
\varepsilon^{2}=\left(\varepsilon_{1}^{1},...,\varepsilon_{|N|-1}^{1},\varepsilon_{|N|}^{2}\right); & \varepsilon^{2}=\left(y_{1}^{1},...,y_{|N|-1}^{1},y_{|N|}^{2}\right)\\
.. & ..\\
\varepsilon^{r_{|N|}}=\left(\varepsilon_{1}^{1},...,\varepsilon_{|N|-1}^{1},\varepsilon_{|N|}^{r_{|N|}}\right); & y^{r_{|N|}}=\left(y_{1}^{1},...,y_{|N|-1}^{1},y_{|N|}^{r_{|N|}}\right)\\
\varepsilon^{r_{|N|}+1}=\left(\varepsilon_{1}^{1},...,\varepsilon_{|N|-1}^{2},\varepsilon_{|N|}^{1}\right); & y^{r_{|N|}+1}=\left(y_{1}^{1},...,y_{|N|-1}^{2},y_{|N|}^{1}\right)\\
\varepsilon^{r_{|N|}+2}=\left(\varepsilon_{1}^{1},...,\varepsilon_{|N|-1}^{2},\varepsilon_{|N|}^{2}\right); & y^{r_{|N|}+2}=\left(y_{1}^{1},...,y_{|N|-1}^{2},y_{|N|}^{2}\right)\\
..
\end{align*}
We then transform $(P1)$ by defining new variables $\tilde{p}=\hat P_{n,x}-q$
and 
\[
z=\left[\begin{array}{c}
z_{1}\\
z_{2}
\end{array}\right]=\left[\begin{array}{c}
\tilde{p}\\
v
\end{array}\right].
\]
As the set of predictions is a subset of the $(|Y|-1)$-dimensional
simplex, we modify the objective of the program to $(\tilde{b},0)^\top\left(\hat P_{n,x}-q\right)$,
where $\tilde{b}$ is a vector in the $\left(|Y|-1\right)-$dimensional
closed ball $\balln$. This modified objective yields
a value of zero if and only if the original program has a value of
zero. The transformed program is:
\begin{eqnarray*}
\max_{\tilde{b}\in\mathbb{R}^{|Y|-1}}\min_{z_{1}\in\mathbb{R}^{|Y|},z_{2}\in\mathbb{R}_{+}^{d_{\nu}}} & \left[\begin{array}{c}
\tilde{b}\\
0_{d_{\nu}+1}
\end{array}\right]^\top z, & (P2)\\
s.t. & \tilde{b}^\top W_{n,x}\tilde{b} & \leq1\\
 & A_{eq}z & =a\\
 & A_{ineq}z & \leq0_{d_{ineq}},
\end{eqnarray*}
where $A_{eq},A_{ineq}$ and are matrices that stack, respectively,
linear equality constraints and linear inequalities, and whose numbers
of rows are $d_{eq}=|Y|+r+1$ and $d_{ineq}=\sum_{i\in N}(|Y_{i}|\cdot|Y_{i}-1|\cdot r_{i})$
. The object $a$ is a vector of constants, and we use $0_{d}$, $1_{d}$
and $I_{d}$ to denote the $d-$vector of zeros and ones, and the
$d\times d$ identity matrix. To construct the matrix $A_{eq},$ notice
that the equality constraints in $(P1)$ can be written as 
\begin{align*}
I_{|Y|}z_{1}+ & A_{eq}^{1}z_{2}=\hat P_{n,x}\\
A_{eq}^{2}z_{2} & =f^{r}\left(\theta_{\varepsilon}\right)\\
1_{d_{v}}^\top z_{2} & =1,
\end{align*}
where $A_{eq}^{1}$ is a matrix with $|Y|$ rows whose $k-$th row
is made of $r$ copies of a $|Y|-$vector with components all zeros
except for the $k-$th component that is equal to one, and $A_{eq}^{2}$
is a block-diagonal matrix with $r$ rows and $1_{|Y|}^\top$ on the
diagonal. The $d_{eq}\times d_{z}$ matrix $A_{eq}$ is then 
\[
A_{eq}=\text{\ensuremath{\left[\begin{array}{cc}
 I_{|Y|}  &  A_{eq}^{1}\\
 0_{\left(r\cdot|Y|\right)}  &  A_{eq}^{2}\\
 0_{|Y|}^\top  &  1_{d_{v}}^\top 
\end{array}\right]}}
\]
with $d_{z}=|Y|\cdot(r+1);$ $a$ is a $d_{eq}-$vector defined as
\[
a=\left[\begin{array}{c}
\hat P_{n,x}\\
f^{r}\left(\theta_{\varepsilon}\right)\\
1
\end{array}\right].
\]
The incentive compatibility inequality constraints in $(P1)$ are
also linear, so that the matrix $A_{ineq}$ can be constructed in
a similar way.

\textbf{Step 3 - Duality and Maximization Program: }Although $(P2)$
is in the form of a maxmin problem, it can be transformed into a maximization
problem by considering the dual of the inner minimization: 
\begin{eqnarray*}
\max_{\tilde{b}\in\mathbb{R}^{|Y|-1},\lambda_{eq}\in\mathbb{R}^{d_{eq}},\lambda_{ineq}\in\mathbb{R}_{+}^{d_{ineq}}} & -\left[\begin{array}{c}
a\\
0_{d_{ineq}}
\end{array}\right]^\top\left[\begin{array}{c}
\lambda_{eq}\\
\lambda_{ineq}
\end{array}\right] & (P3)\\
s.t. & \tilde{b}^\top W_{n,x}\tilde{b} & \leq1\\
 & \left(A^\top\right)_{1:|Y|}\left[\begin{array}{c}
\lambda_{eq}\\
\lambda_{ineq}
\end{array}\right] & =-\left[\begin{array}{c}
\tilde{b}\\
0
\end{array}\right]\\
 & \left(A^\top\right)_{|Y|+1:d_{z}}\left[\begin{array}{c}
\lambda_{eq}\\
\lambda_{ineq}
\end{array}\right] & \geq0_{d_{\nu}},
\end{eqnarray*}
where $A=\left[\begin{array}{c}
A_{eq}\\
A_{ineq}
\end{array}\right],$ the vectors $\lambda_{eq}$ and $\lambda_{ineq}$ are the dual variables
associated to the constraints of $\left(P2\right),$ and $\left(A^\top\right)_{1:|Y|}$
and $\left(A^\top\right)_{|Y|+1:d_{z}}$ denote the first $|Y|$ and
the last rows of the matrix $A^\top$ and $d_{A}=d_{eq}+d_{ineq}$
is the number of rows of $A$. By strong duality, as well as by the existence
of BCE, $(P3)$ has the same value than $(P2)$. 

Finally, let $w=(\tilde b^\top,\lambda_{eq}^\top,\lambda_{ineq}^\top)^\top$, $\gamma=(0_{|Y|-1}^\top,a^\top,0^\top_{d_{ineq}})^\top$, and
\begin{align*}
    \Gamma_1=\begin{bmatrix}
        W_{n,x}&0_{|Y|-1\times d_{A}}\\
        0_{d_Z\times |Y|-1}&0_{d_A\times d_A}
    \end{bmatrix},~
    \Gamma_2=\begin{bmatrix}
        I_{|Y|-1}  &    \multirow{2}{*}{$\left(A^\top\right)_{1:|Y|}$}\\
        1_{|Y|}^\top &
    \end{bmatrix},~
    \Gamma_3=\begin{bmatrix}
        0_{d_\nu\times |Y|-1} & \left(A^\top\right)_{|Y|+1:d_{z}}
    \end{bmatrix}.
\end{align*}
Then one may represent $(P3)$ as in \eqref{eq:def_Vnx}.

\subsection{Computing $\sup_\theta p_n(\theta)$}
 Below, we drop the subscript $n$ from $p_n$ as it does not play a role.
The following algorithm yields a sequence of tentative optimal values $p(\theta^*_L),L=k+1,k+2,\dots$, which tends to $\sup_{\theta\in\Theta}p(\theta)$.
\begin{description}
\item[Step 1]: Draw randomly (uniformly) over $\Theta$ a set $(\theta^{(1)},\dots,\theta^{(k})$ of initial evaluation points. Evaluate $p(\theta^{(\ell)}),\ell= 1,\dots,k-1$. Initialize $L = k$.
\item[Step 2]: Record the tentative optimal value
\begin{align}
    p(\theta^*_L)=\max\{p(\theta^{(\ell)}),\ell \in \{1,\dots,L\}\}.
\end{align}
If $p(\theta^*_L)>\alpha$, halt the algorithm. Set $\phi=0$ (i.e. do not reject $H_0$). Otherwise, proceed to Step 3.
\item[Step 3]: Approximate $\theta\mapsto p(\theta)$ by a flexible auxiliary model. A commonly used choice is a Gaussian-process regression, which models each observed value $Y^{(\ell)}=p(\theta^{(\ell)})$ as 
\begin{align}
    Y^{(\ell)}=\mu+\zeta(\theta^{(\ell)}),~\ell=1,\dots, L,
\end{align}
where $\zeta$ is a mean-zero Gaussian process indexed by $\theta$ with constant variance and a correlation 
\begin{align}
    \text{Corr}(\zeta(\theta),\zeta(\theta'))=K_\beta(\theta-\theta')
\end{align}
for some kernel function $K_\beta$ (e.g. $K_\beta=\exp(-\sum_{h=1}^d(\theta_h-\theta_h')^2/\beta_h)$). The unknown parameters can be estimated by running a feasible-GLS regression of $\mathbf Y=(Y^{(1)},\dots,Y^{(L)})$ on a constant with the given correlation matrix.
 The (best linear) predictor of $p(\theta)$ then has a closed form\footnote{Its derivative also has a closed form. See \cite{jones1998efficient}.}
\begin{align}
    p_L(\theta)=\hat\mu+\mathbf{r}_L(\theta)^{\top}\mathbf{R}_L^{-1}(\mathbf Y-\hat\mu \mathbf 1).
\end{align}
This predictor coincides with $p$ at the evaluation points (i.e., $p_L(\theta^{(\ell)})=p(\theta^{((\ell))}),\ell=1,\dots,L$) providing an analytical interpolation. Also, the uncertainty left in $p$ is quantified by the following variance:
\begin{align}
\hat\varsigma^2s^2_L(\theta)= \hat\varsigma^2\left(1-\mathbf r_L(\theta)^{\top}\mathbf R_L^{-1}\mathbf r_L(\theta)+\frac{(1-\mathbf 1^{\top}\mathbf R_L^{-1}\mathbf r_L(\theta))^2}{\mathbf 1^{\top}\mathbf R_L^{-1}\mathbf 1}\right).
\end{align}
\item[Step 4:] With probability $1-\epsilon$, obtain the next evaluation point $\theta^{(L+1)}$ as
	\begin{align}
		\theta^{(L+1)}\in\argmax_{\theta\in\Theta} (p_L(\theta)-p(\theta^*_{L}))\Phi\Big(\frac{p_L(\theta)-p(\theta^*_{L})}{\hat\varsigma s_L(\theta)}\Big)+\hat\varsigma s_L(\theta)\phi\Big(\frac{p_L(\theta)-p(\theta^*_{L})}{\hat\varsigma s_L(\theta)}\Big),
	\end{align}
	where the objective function is called the \emph{expected improvement function}.
\end{description}
Set $L\leftarrow L+1$ and return to Step 1. Repeat the steps until convergence. The algorithm can be interpreted as follows.
This algorithm first evaluates $p$ on a coarse grid (Step 1). It then approximates $p$ by a tractable Gaussian process regression model, which can be used to guide the determination of the next evaluation point to draw. For this, we need to take into account where the maximum is likely to lie (exploitation) and where the current approximation is rough and benefit from additional evaluation points (exploration). The expected improvement function in Step 4 provides a criterion for how to trade them off optimally. Repeating this process generates a sequence of that tends to the global maximum of $p$. This can be implemented by open softwares such as DACE (in Matlab) and MOE (in Python).\footnote{These softwares are available from \url{http://www.omicron.dk/dace.html} (DACE) and \url{https://github.com/wujian16/Cornell-MOE} (MOE) respectively.}

\section{Proofs}

\subsection{Definitions and Notation}\label{ssec:notation}
Below, we introduce objects that will be used in the proof of auxiliary lemmas. In Lemma \ref{lem:size_control} below, we consider a sequence $\eta_n$ such that $\tau_n^{-1}\eta_n\to\pi \in\mathbb R^{K\times X}_{-,\infty}$.
 For $\xi=\tau_n^{-1}\eta_n(b,x)$, define
\begin{align}
\varphi(\xi) & \equiv\begin{cases}
0 & \text{if }\xi\ge-1\\
-\infty & \text{if }\xi<-1~,
\end{cases}
\qquad \varphi^{*}(\xi)  \equiv\begin{cases}
\varphi(\xi) & \text{if }\pi=0\\
-\infty & \text{if }\pi<0~.
\end{cases}\label{eq:phi^star}
\end{align}
Recall that $W_x$ is the probability limit of $W_{n,x}$ (see Assumption \ref{as:Wn}). 
Let $K\subset \mathbb R^{|Y|}$ be a compact set such that $\ball\subset K^{-\epsilon}$ uniformly across DGPs for some uniform constant $\epsilon>0$, where $K^{-\epsilon}=\{x\in K:\inf_{y\in K^c}\|x-y\|\ge \epsilon\}$. We then define an empirical process $\mathbb G_n$ (indexed by $( b,x)\in K\times X$) and bootstrapped empirical process $\mathbb G^*_n$ as follows:
\begin{align}
	\mathbb{G}_{n}( b,x)\equiv\sqrt{n} b^{\top}(\hat{P}_{n,x}-P_{y|x})~,\qquad \mathbb{G}^*_{n}( b,x)\equiv\sqrt{n} b^{\top}(\hat{P}^*_{n,x}-\hat P_{n,x}).
\end{align}
Under our assumptions $\mathbb G_n$ converges weakly (in the sense of Hoffmann-J{\o}rgensen) to a tight Gaussian process $\mathbb G$ \citep{van_der_Vaart_1996}.
For any $\pi(b,x)\in[-\infty,0]$, let 
\begin{align}
\pi^{*}(b,x)=\begin{cases}
0 & \text{if }\pi(b,x)=0\\
-\infty & \text{if }\pi(b,x)<0.
\end{cases}\label{eq:pi_star}
\end{align}
Define $c_{\pi^{*}}(1-\alpha)$ be the $1-\alpha$ quantile of
\begin{align}
\sup_{(b,x)\in\Psi_{\infty}}\mathbb{G}(b,x),
\end{align}
where 
\begin{align}
\Psi_{\infty}=\{(b,x)\in K\times X:\pi^{*}(b,x)=0,~b\in \ball\}.
\end{align}
This is the set of constraints that will be selected by the GMS asymptotically.

Let $\mathcal F\equiv\{(\theta,P):P\in\mathcal P_\theta,~\theta\in\Theta\}$ be the set of $(\theta,P)$ pairs that are compatible with our assumptions.
Following \cite{Andrews:2010aa}, we introduce a one-to-one mapping between $(\theta,P)\in \mathcal F$ and a new parameter $\gamma=(\gamma_1,\gamma_2,\gamma_3)$ with corresponding parameter space $\Gamma.$ For each $(b,x)\in K\times X$, $\gamma_1\in\mathbb R_{-}^{K\times X}$ is defined by the relation
\begin{align}
	b^{\top}P_{y|x}-h(b,Q_{\theta,S^{r}}^{BCE}(x))-\gamma_1(b,x)=0.\label{eq:gamma1}
\end{align}
For each $x\in X$, let $W_x=\text{AsyVar}_P(\sqrt n \hat P_{n,x})$, and we let $\gamma_2=(\theta,vech(W_{x_1}),\dots,vech(W_{x_{|X|}}))^{\top}.$ Finally, we let $\gamma_3=P.$
We then define $\Gamma$ as 
\begin{multline}
	\Gamma\equiv \big\{\gamma=(\gamma_1,\gamma_2,\gamma_3):\text{ for some }(\theta,P)\in \mathcal F, \gamma_1 \text{ satisfies \eqref{eq:gamma1}, }\\ \gamma_2=(\theta,vech(W_{x_1}),\dots,vech(W_{x_{|X|}}))^{\top},~\gamma_3=P\big\}.
\end{multline}
Note that, since we impose $(\theta,P)\in \mathcal F$, all parameters in $\Gamma$ must respect Assumptions \ref{as:Model}-\ref{as:px}.
In what follows, for any sequence $\{\gamma_{n}\}$, we let
\begin{align}
	\eta_n(b,x)\equiv \sqrt n\gamma_{1,n}(b,x).
\end{align}
Other notation and definitions used throughout are collected in the following table.
\begin{table}[h]
    \centering
    \begin{tabular}{rl}
        $a\lesssim b$   & $a\le Mb$ for some constant $M$.  \\
        $\|\cdot\|_{op}$ & the operator norm for linear mappings. \\
        $\|\cdot\|_{\mathcal F}$ & the supremum norm over $\mathcal F$.\\
        $\|\cdot\|_{\psi_{2}}$ & the sub-Gaussian norm: $\|X\|_{\psi_2}=\inf\{t>0:E[\exp(X^{2}/t^{2})]\le2\}$.\\
        $N(\epsilon,\mathcal F,\|\cdot\|)$ & covering number of size $\epsilon$ for $\mathcal F$ under norm $\|\cdot\|$. \\
        $N_{[]}(\epsilon,\mathcal F,\|\cdot\|)$ & bracketing number of size $\epsilon$ for $\mathcal F$ under norm $\|\cdot\|$.\\
        $X_n\stackrel{P^n}{\leadsto}X$ & $X_n$ weakly converges to $X$ under $\{P^n\}$. 
    \end{tabular}
\end{table}

\subsection{On Studentization}
\begin{lemma}\label{lem:studentization}
Let $\balln=\{b\in\mathbb R^{|Y|}:\wnormn{b}\le 1\}.$ Then \eqref{eq:studentization} holds.
\end{lemma}
\begin{proof}
Let $\mathbb{S}_n=\{b:\wnormn{b}=1\}$ be the set of unit vectors with respect to $\wnormn{\cdot}$.
Note that the map $b\mapsto b^{\top}\hat{P}_{n,x}-h(b,\bce{S^r})$ is positively homogeneous  because the support function is positively homogeneous \citep[][Appendix F]{Molchanov:2005aa}. One may then write 
\begin{multline}
	\sup_{ b\in \balln}\sqrt{n}\{ b^{\top}\hat{P}_{n,x}-h( b,\bce{S^r})\}=\sup_{\gamma\in [0,1]}\gamma\sup_{\bar b\in \mathbb{S}_n}\sqrt n\{ \bar b^{\top}\hat{P}_{n,x}-h(\bar b,\bce{S^r})\}\\
	=\begin{cases}
	\sup_{\bar b\in \mathbb{S}_n}\sqrt n\{ \bar b^{\top}\hat{P}_{n,x}-h(\bar b,\bce{S^r})\}&\text{if } \sup_{\bar b\in \mathbb{S}_n}\sqrt n\{ \bar b^{\top}\hat{P}_{n,x}-h(\bar b,\bce{S^r})\}>0\\
		0&\text{if } \sup_{\bar b\in \mathbb{S}_n}\sqrt n\{ \bar b^{\top}\hat{P}_{n,x}-h(\bar b,\bce{S^r})\}\le 0,
	\end{cases}\label{eq:truncation}
\end{multline}
where the second equality holds because it is optimal to set $\gamma=0$ whenever the optimal value of the inner maximization problem is non-positive, and it is optimal to set $\gamma=1$ otherwise.
Finally, note that one may represent $\bar b$ as $\bar b=b/(b^{\top}W_{n,x}b)^{1/2}$ for  a nonzero vector $b$.  Therefore, again using the positive homogeneity, we may write
\begin{align}
 \sup_{b\in \balln}\sqrt n\{b^{\top}\hat P_{n,x}-h(b,\bce{S^r})\}=   \sup_{b\in \mathbb R^{|Y|}\setminus\{0\}}	\sqrt n\Big\{ \frac{b^{\top}\hat{P}_{n,x}-h(b,\bce{S^r})}{\sqrt{nb^{\top}W_{n,x}b}}\Big\}_+.\
\end{align}
\end{proof}

\subsection{Results and Auxiliary Lemmas on Hypothesis Testing}

\begin{proof}[Proof of Theorem \ref{thm:size_control}]	
We first observe that rejection based on the $p$-value can be restated as follows:
\begin{align}
p_{n}(\theta)<\alpha,~\forall\theta\in\Theta,~~~\Leftrightarrow~~~T_{n}(\theta)>\hat{c}_{n}(\theta,1-\alpha),~\forall\theta\in\Theta,
\end{align}
where $\hat{c}_{n}(\theta,1-\alpha)=\inf\{c:P^{*}( T_{n}^{*}(\theta)\le c|y^{n},x^{n})\ge 1-\alpha\}.$
Therefore, for any fixed $\theta\in\Theta$, 
\begin{align}
P\big(\sup_{\tilde\theta\in\Theta}p_{n}(\tilde\theta)<\alpha\big)\le P( T_{n}(\theta)>\hat{c}_{n}(\theta,1-\alpha)).
\end{align}
Hence,
\begin{align}
	\sup_{P\in\mathcal P}E_P[\phi]\le \sup_{(\theta,P)\in\mathcal F}P( T_{n}(\theta)>\hat{c}_{n}(\theta,1-\alpha)).
\end{align}
 Let  $\{(P_n,\theta_n)\in\mathcal P\times \Theta\}_{n=1}^\infty$ be a sequence such that
\begin{align}
	\limsup_{n\to\infty}P_n( T_{n}(\theta_n)>\hat{c}_{n}(\theta_n,1-\alpha))=\limsup_{n\to\infty}\sup_{(\theta,P)\in\mathcal F}P( T_{n}(\theta)>\hat{c}_{n}(\theta,1-\alpha)).
\end{align}
Let $\text{RP}_n(\gamma_n)=P_n( T_{n}(\theta_n)>\hat{c}_{n}(\theta_n,1-\alpha)),$ where $\gamma_n$ is a sequence of parameters defined as in Section \ref{ssec:notation}.
Let $\{u_n\}$ be a subsequence of $\{n\}$ such that $\lim_{n\to\infty}\text{RP}_{u_n}(\gamma_{u_n})$ exists and
\begin{align}
	\lim_{n\to\infty}\text{RP}_{u_n}(\gamma_{u_n})=\limsup_{n\to\infty}\text{RP}_{n}(\gamma_{n}).
\end{align}
Such a sequence exsits without loss of generality. Following (S1.1)-(S1.4) in \cite{Andrews:2010aa} (supplementary material), it is straightforward to constrcut a further subsequence $\{w_n\}$ of $\{u_n\}$ such that (i) $\eta_{w_n}=\sqrt{w_n}\gamma_{1,w_n}\to \eta\in \mathbb R^{K\times X}_{-,\infty}$, (ii) $\tau^{-1}_{w_n}\eta_{w_n}\to\pi \in\mathbb R^{K\times X}_{-,\infty}$, and $\gamma_{2,w_n}\to\gamma_2\in \text{cl}(\Gamma_2)$.
 By Lemma \ref{lem:size_control}, it then follows that $\lim_{n\to\infty}\text{RP}_{w_n}(\gamma_{w_n})\le\alpha$. Hence, the conclusion of the theorem holds.
\end{proof}

\begin{lemma}\label{lem:size_control} Suppose Assumptions \ref{as:primitives}-\ref{as:px}
hold. Let $\{\gamma_{n}:n\ge1\}$ be a sequence in $\Gamma$
where the components $\eta_{n}$ and $(W_{n,x},x\in X)$ of $\gamma_{n,\eta}$
satisfy (i) $\eta_{n}\to\eta\in\mathbb{R}_{-,\infty}^{K \times X}$
(pointwise in $(b,x)$); (ii) $\tau_{n}^{-1}\eta_{n}\to\pi\in\mathbb{R}_{-,\infty}^{K\times X}$;
(iii) $\gamma_{2,n}\to\gamma\in \mathbb R^{d_{\gamma_2}}_{\pm \infty}$, where $d_{\gamma_2}$ is the dimension of $\gamma_2$. Then, (a)-(c) hold.

(a) $\hat{c}_{n}(\theta_{n},1-\alpha)\ge c_{n}^{*},a.s.$ for all
$n$ for a sequence of random variables $\{c_{n}^{*}\}$ such that
$c_{n}^{*}\stackrel{p}{\to}c_{\pi^{*}}(1-\alpha)$.

(b) $\limsup_{n\to\infty}P_{\eta_{n}}(\tilde T_{1,n}(\theta_{n})>\hat{c}_{n}(\theta_{n},1-\alpha))\le\alpha.$

(c) Parts (a) and (b) hold with all subsequences $\{w_{n}\}$ in place
of $\{n\}$ provided that (i)--(iii) hold with $w_{n}$ in place
of $n$. \end{lemma}

\begin{proof}
Rewrite $T_n$ and $T_n^{*}$ as 
\begin{align*}
T_n(\theta) & =\sup_{x\in X}\sup_{b\in \balln}\{\mathbb{G}_{n}(b,x)+\eta_{n}(b,x)\},\\
T_n^{*}(\theta) & =\sup_{x\in X}\sup_{b\in \balln}\{\mathbb{G}_{n}^{*}(b,x)+\varphi(\hat{\xi}_{n}(b,x;\theta))\},
\end{align*}
where $\hat{\xi}_{n}(b,x;\theta) \equiv\tau_{n}^{-1}\hat{\eta}_{n}(b,x;\theta).$
Note that $c_{\pi^{*}}(1-\alpha)$ is the $1-\alpha$ quantile of
\begin{align*}
 & \sup_{x\in X}\sup_{b\in \ball}\{\mathbb{G}(b,x)+\pi^{*}(b,x)\}.
\end{align*}
To prove part (a), first, when $c_{\pi^{*}}(1-\alpha)=0$, define
$c_{n}^{*}=0$ so that 
\begin{align*}
\hat{c}_{n}(\theta_{n},1-\alpha) & \ge c_{n}^{*}\overset{p}{\rightarrow}c_{\pi^{*}}(1-\alpha)
\end{align*}
trivially holds since $T_n^{*}(\theta)\ge0$ for all $\theta\in\Theta$.

Next, suppose $c_{\pi^{*}}(1-\alpha)>0$. 
By construction, 
\begin{align}
\varphi^{*}(\hat{\xi}_{n}(b,x;\theta_{n})) & \le\varphi(\hat{\xi}_{n}(b,x;\theta_{n}))\text{ a.s. }\forall n\text{ and }(b,x).\label{eq:compare_phis}
\end{align}
Let $c_{n}^{*}$ be the $1-\alpha$ quantile of 
\begin{align*}
 & \sup_{x\in X}\sup_{b\in\balln}\{\mathbb{G}_{n}^{*}(b,x)+\varphi^{*}(\hat{\xi}_{n}(b,x;\theta))\}.
\end{align*}
Then $\hat{c}_{n}(\theta_{n},1-\alpha)\ge c_{n}^{*}$ a.s. $\forall n$
by \eqref{eq:compare_phis}.

We now show $c_{n}^{*}\overset{p}{\rightarrow}c_{\pi^{*}}(1-\alpha)>0$. Let $P_{M_n}$ denote the probability measure governing the (multinomial) bootstrap weights, and let $E_{M_n}$ be the associated expectation operator. Let $BL_1$ denote the class of bounded Lipschitz functions on $K\times X$ with Lipschitz constant equal to one.
By Lemma \ref{lem:uniform_entropy} and arguing as in the proof of Lemma \ref{lem:boot_se}, $\mathbb G^*_n$ is a bootstrapped empirical process defined on a uniform Donsker class.
Therefore, applying Theorem 3.6.2 in \cite{van_der_Vaart_1996} under any
$\{P_n\}\subset\mathcal P$, we have 
\begin{align}
\sup_{h\in BL_{1}}|E_{M_{n}}[h(\mathbb{G}_{n}^{*})|(y^{n},x^{n})]-E[h(\mathbb{G})]|\stackrel{P_{n}}{\to}0.\label{eq:gaussian}
\end{align}
For each $(b,x)$, 
\begin{align}
\hat{\xi}_{n}(b,x;\theta_{n}) & =\tau_{n}^{-1}\sqrt{n}(b^{\top}\hat{P}_{n,x}-h(b,Q_{\theta_{n},S^{r}}^{BCE}(x)))\nonumber \\
 & =\tau_{n}^{-1}\sqrt{n}(b^{\top}P_{n,y|x}-h(b,Q_{\theta_{n},S^{r}}^{BCE}(x)))+\tau_{n}^{-1}\sqrt{n}b^{\top}(\hat{P}_{n,x}-P_{n,y|x})\nonumber \\
 & =\tau_{n}^{-1}\eta_{n}(b,x)+o_{p^{*}}(1)\nonumber \\
 & \overset{p^{*}}{\rightarrow}\pi(b,x),\label{eq:pw_conv_0}
\end{align}
where the convergence is by (ii) and the second equality is by Assumption
3 that guarantees $\hat{P}_{nx}-P_{n,y|x}=O_{p}(n^{-1/2})$ for each
$x$, which in turn implies $\tau_{n}^{-1}\sqrt{n}(b^{\top}\hat{P}_{n,x}-P_{n,y|x})=o_{p^{*}}(1)$,
and for r.v.'s $z_{n}$ and $z$, the convergence $z_{n}\overset{p^{*}}{\rightarrow}z$
is defined as $P^{*}(\left|z_{n}-z\right|>\varepsilon|y^{n},x^{n})\overset{p}{\rightarrow}0$.

We now show that $\varphi^{*}(\xi(b,x))\rightarrow\varphi^{*}(\pi(b,x))$
for any sequence $\xi(b,x)$ for which $\xi(b,x)\rightarrow\pi(b,x)$.
If $\pi(b,x)=0$, then 
\begin{align*}
\varphi^{*}(\xi(b,x)) & =\varphi(\xi(b,x))\rightarrow\varphi(\pi(b,x))=\varphi^{*}(\pi(b,x))=0
\end{align*}
as $\xi(b,x)\rightarrow\pi(b,x)$, by the definitions in \eqref{eq:phi^star}. If $\pi(b,x)<0$, then 
\begin{align*}
\varphi^{*}(\xi(b,x)) & =-\infty=\varphi^{*}(\pi(b,x))
\end{align*}
by \eqref{eq:phi^star}. Therefore, as $\xi(b,x)\rightarrow\pi(b,x)$,
we have $\varphi^{*}(\xi(b,x))\rightarrow\varphi^{*}(\pi(b,x))=\pi^{*}(b,x)$
where the last equality is by the definitions in  \eqref{eq:phi^star} and \eqref{eq:pi_star}.

Let $\delta>0.$ Define the following sets of sample paths: 
\begin{align}
A_{1n} & \equiv\Big\{(y^{n},x^{n}):\inf_{(b,x):\pi(b,x)=0}\hat{\xi}_{n}(b,x)<-1\Big\}\label{eq:def_A1n}\\
A_{2n} & \equiv\Big\{(y^{n},x^{n}):\sup_{x\in X} d_H(\balln,b_x)>\delta\Big\}.\label{eq:def_A2n}
\end{align}
Lemmas \ref{lem:maximal} and \ref{lem:Hausdorff_consistency} show that $P_{n}(A_{jn})\to0,$ for $j=1,2$.

Now note that, for any $a\in\mathbb{R}$, 
\begin{multline}
\sup_{x\in X}\sup_{b\in\balln}\{\mathbb{G}_{n}^{*}(b,x)+\varphi^{*}(\hat{\xi}_{n}(b,x;\theta_{n}))\}\le a\\
\Leftrightarrow\sup_{(b,x)\in \Psi_{n}}\{\mathbb{G}_{n}^{*}(b,x)+\varphi^{*}(\hat{\xi}_{n}(b,x;\theta_{n}))\}\le a\text{ and }\sup_{(b,x)\in\Psi_n^c}\{\mathbb{G}_{n}^{*}(b,x)+\varphi^{*}(\hat{\xi}_{n}(b,x;\theta_{n}))\}\le a,\label{eq:sup_emp+slack}
\end{multline}
where $\Psi_{n}=\{(b,x)\in K\times X:\pi^{*}(b,x)=0,~ b\in\balln\}$ and $\Psi_n^c=\{(b,x)\in K\times X:\pi^{*}(b,x)\ne 0,~ b\in\balln\}$.

For $(b,x)\in\Psi_n^c$, we have previously shown that $\varphi^{*}(\hat{\xi}_{n}(b,x;\theta_{n}))=-\infty$,
and thus $\sup_{(b,x)\in\Psi_n^c}\{\mathbb{G}_{n}^{*}(b,x)+\varphi^{*}(\hat{\xi}_{n}(b,x;\theta_{n}))\}=0$. Therefore, by \eqref{eq:sup_emp+slack}, one may write
\begin{align}
P_{M_{n}}&\Big(\sup_{x\in X}\sup_{b\in\balln}\{\mathbb{G}_{n}^{*}(b,x)+\varphi^{*}(\hat{\xi}_{n}(b,x;\theta_{n}))\}\le a\big|y^{n},x^{n}\Big)\notag\\
= & P_{M_{n}}\Big(\sup_{(b,x)\in\Psi_n}\{\mathbb{G}_{n}^{*}(b,x)+\varphi^{*}(\hat{\xi}_{n}(b,x;\theta_{n}))\}\le a,0\le a\big|y^{n},x^{n}\Big)\notag\\
= & \begin{cases}
P_{M_{n}}\Big(\sup_{(b,x)\in\Psi_n}\{\mathbb{G}_{n}^{*}(b,x)+\varphi^{*}(\hat{\xi}_{n}(b,x;\theta_{n}))\}\le a\big|y^{n},x^{n}\Big) & \text{if }a\ge0,\\
0 & \text{if }a<0,
\end{cases}\label{eq:cdf_Gstar}.
\end{align}
Similarly, for any $a\in \mathbb R$,
\begin{multline}
	\sup_{x\in X}\sup_{b\in\ball}\{\mathbb{G}(b,x)+\pi^{*}(b,x)\}\le a\\
	\Leftrightarrow \sup_{(b,x)\in\Psi_\infty}\{\mathbb{G}(b,x)+\pi^{*}(b,x)\}\le a\text{ and }\sup_{(b,x)\in\Psi^c_\infty}\{\mathbb{G}(b,x)+\pi^{*}(b,x)\}\le a.
\end{multline}
Observe that $\pi^{*}(b,x)-\infty$ for any $(b,x)\in \Psi_\infty^c$ implying $\sup_{(b,x)\in\Psi_\infty^c}\{\mathbb{G}(b,x)+\pi^{*}(b,x)\}=0$. Therefore, by mimicking the argument in \eqref{eq:cdf_Gstar},
\begin{align}
 P\Big(\sup_{x\in X}\sup_{b\in\ball}\{\mathbb{G}(b,x)+\pi^{*}(b,x)\}\le a\Big)
= & \begin{cases}
P\Big(\sup_{(b,x)\in\Psi_\infty}\{\mathbb{G}(b,x)+\pi^{*}(b,x)\}\le a\Big) & \text{if }a\ge0,\\
0 & \text{if }a<0
\end{cases}\notag\\
= & \begin{cases}
P\Big(\sup_{(b,x)\in\Psi_\infty}\{\mathbb{G}(b,x)\}\le a\Big) & \text{if }a\ge0,\\
0 & \text{if }a<0,
\end{cases}
\label{eq:cdf_G}
\end{align}
where the last equality follows from $\pi^{*}(b,x)=0$ on $\Psi_\infty.$

Now, for each $a\in\mathbb R$, let 
\begin{multline}
r_{n}(a)\equiv P_{M_{n}}\Big(\sup_{x\in X}\sup_{b\in\balln}\{\mathbb{G}_{n}^{*}(b,x)+\varphi^{*}(\hat{\xi}_{n}(b,x;\theta_{n}))\}\le a|(y^{n},x^{n})\Big)\\
-P\Big(\sup_{x\in X}\sup_{b\in\ball}\{\mathbb{G}(b,x)+\pi^{*}(b,x)\}\le a\Big).
\end{multline}

 Suppose $(y^{n},x^{n})\in A_{1n}^{c}\cap A_{2n}^c$. On this event,
$\varphi^{*}(\hat{\xi}_{n}(b,x;\theta_{n}))=0$ uniformly on $\Psi_n$, and thus $\sup_{\Psi_n}\{\mathbb{G}_{n}^{*}(b,x)+\varphi^{*}(\hat{\xi}_{n}(b,x;\theta_{n}))\}=\sup_{\Psi_n}\{\mathbb{G}_{n}^{*}(b,x)\}$. This, together with \eqref{eq:cdf_Gstar} and \eqref{eq:cdf_G} imply that, for any given $a>0$
 and any $\epsilon>0$, 
\begin{multline}
	P_{n}(\{|r_{n}(a)|>\epsilon\}\cap A_{1n}^{c}\cap A_{2n}^c)\\
	=P_{n}\Big(\Big\{\big|P_{M_n}\big(\sup_{(b,x)\in\Psi_n}\{\mathbb{G}_{n}^{*}(b,x)\}\le a\big)-P\big(\sup_{(b,x)\in\Psi_\infty}\{\mathbb{G}(b,x)\}\le a\big)\big|>\epsilon\Big\} \cap A_{1n}^{c}\cap A_{2n}^c\Big)\to 0,
\end{multline}
where the convergence to 0 follows from Lemma \ref{lem:weak_convergence}. Finally, 
By Lemma \ref{lem:maximal}, we then have 
\begin{align}
P_{n}(|r_{n}(a)|>\epsilon)\le P_{n}(\{|r_{n}(a)|>\epsilon\}\cap  A_{1n}^{c}\cap A_{2n}^c)+P_{n}(A_{1n}\cup A_{2n})\to0.
\end{align}

Since this holds for all $a$ in the neighborhood of $c_{\pi^{*}}(1-\alpha)>0$,
by Lemma 5 of \cite{andrews2010asymptotic}, we have 
\begin{align*}
c_{n}^{*} & \overset{p}{\rightarrow}c_{\pi^{*}}(1-\alpha).
\end{align*}

We now prove part (b). First, if $\pi(b,x)<0$, then $\eta(b,x)=\pi^{*}(b,x)=-\infty$
by conditions (i) and (ii), and if $\pi(b,x)=0$, then $\eta(b,x)\le0$
and $\pi^{*}(b,x)=0$. That is, for each $(b,x)$, we have $\pi^{*}(b,x)\ge\eta(b,x)$,
and therefore 
\begin{align*}
\sup_{x\in X}\sup_{b\in\ball}\{\mathbb{G}(b,x)+\pi^{*}(b,x)\}_{+}\ge & \sup_{x\in X}\sup_{b\in\ball}\{\mathbb{G}(b,x)+\eta(b,x)\}_{+}
\end{align*}
and thus 
\begin{align}
c_{\pi^{*}}(1-\alpha) & \ge c_{\eta}(1-\alpha),\label{eq:c_eta}
\end{align}
where $c_{\eta}(1-\alpha)$ denotes the $1-\alpha$ quantile of $\sup_{x\in X}\sup_{b\in\ball}\{\mathbb{G}(b,x)+\eta(b,x)\}$.

Next, under $\{\gamma_{n,\eta}:n\ge1\}$, $\mathbb{G}_{n}\rightsquigarrow\mathbb{G}$
and $\eta_{n}\rightsquigarrow\eta$ by a parallel argument to that
for the bootstrap weak convergence above. Then, by the continuous
mapping theorem, we have 
\begin{align}
T_{n}(\theta_{n}) & \overset{d}{\rightarrow}J_{\eta},\label{eq:J}
\end{align}
where $J_{\eta}$ is the distribution of $\sup_{(b,x)\in B\times X}\{\mathbb{G}(b,x)+\eta(b,x)\}_{+}$.
We then have 
\begin{align*}
\limsup_{n\rightarrow\infty}P_{\gamma_{n,\eta}}(T_{n}(\theta_{n})>\hat{c}_{n}(\theta_{n},1-\alpha)) & \le\limsup_{n\rightarrow\infty}P_{\gamma_{n,\eta}}(T_{n}(\theta_{n})>c_{n}^{*})\le1-J_{\eta}(c_{\pi^{*}}(1-\alpha)-),
\end{align*}
where the first inequality is by part (a) that $\hat{c}_{n}(\theta_{n},1-\alpha)\ge c_{n}^{*}$
and the second inequality is by part (a) that $c_{n}^{*}\overset{p}{\rightarrow}c_{\pi^{*}}(1-\alpha)$
and \eqref{eq:J} with $J_{\eta}(a-)$ being the limit from the left
of $J_{\eta}(\cdot)$ at $a$.

Suppose $c_{\pi^{*}}(1-\alpha)>0$. Then 
\begin{align*}
J_{\eta}(c_{\pi^{*}}(1-\alpha)-) & =J_{\eta}(c_{\pi^{*}}(1-\alpha))\ge1-\alpha,
\end{align*}
where the equality holds by $J_{\eta}(a)$ being continuous $\forall a>0$
and the inequality is by \eqref{eq:c_eta}. This proves part (b) for
this case.

Now suppose $c_{\pi^{*}}(1-\alpha)=0$. This implies that $c_{\eta}(1-\alpha)=0$
by \eqref{eq:c_eta}. Under $\{\gamma_{n,\eta}:n\ge1\}$, we have
\begin{align*}
\limsup_{n\rightarrow\infty}P_{\gamma_{n,\eta}}(T_{n}(\theta_{n})>0) & =1-J_{\eta}(0)=1-J_{\eta}(c_{\eta}(1-\alpha))\le\alpha
\end{align*}
where the first equality is by \eqref{eq:J}. This proves
part (b) for this case.

The proof of part (c) is analogous to that for parts (a) and (b) with
$w_{n}$ in place of $n$.

\end{proof}

Let $I_{x}=(1\{y_{\ell}=y_{1},x_{\ell}=x\},\dots,1\{y_{\ell}=y_{|Y|},x_{\ell}=x\})\in\{0,1\}^{|Y|}$.
Suppose for the moment $p_{x}$ is known. We may then write $\mathbb{G}_{n}(b,x)=-\frac{1}{\sqrt{n}}\sum_{\ell=1}^{n}(Z_{\ell}(b)-E[Z_{\ell}(b)])/p_{x},$
where $Z_{\ell}(b)=-I_{x}'b$.

\begin{lemma}\label{lem:maximal} Suppose Assumptions \ref{as:primitives}-\ref{as:sampling} and \ref{as:px} hold.  Let $\tau_{n}\to\infty$. Then,
for any sequence $\{P_{n}\}\subset\mathcal{P}$, $P_{n}(A_{n})\to0.$
\end{lemma}
\begin{proof} For each $x\in X$, let $\Pi(x)\equiv\{b\in \ball:\pi(b,x)=0\}.$
Fix $b\in\Pi(x)$. Observe that 
\begin{align*}
\hat{\xi}_{n}(b,x)<-1 & ~\Leftrightarrow\tau_{n}^{-1}\hat{\eta}_{n}(b,x)<-1\\
 & ~\Leftrightarrow\tau_{n}^{-1}\eta_{n}(b,x)+\tau_{n}^{-1}\mathbb{G}_{n}(b,x)<-1\\
 & ~\Leftrightarrow\mathbb{G}_{n}(b,x)<-\tau_{n}-\eta_{n}(b,x)\\
 & ~\Leftrightarrow\frac{1}{\sqrt{n}}\sum_{\ell=1}^{n}Z_{\ell}(b)-E[Z_{\ell}(b)]>p_{x}(\tau_{n}+\eta_{n}(b,x))\\
 & ~\Rightarrow\frac{1}{\sqrt{n}}\sum_{\ell=1}^{n}Z_{\ell}(b)-E[Z_{\ell}(b)]>\underline\zeta(\tau_{n}+\eta_{n}(b,x)),
\end{align*}
where the last implication is due to $p_{x}\ge\underline\zeta$ for all $x$.

Let $\tilde{\tau}_{n}\equiv\tau_{n}+\inf_{b\in\Pi(x)}\eta_{n}(b,x)$
and note that $\tilde{\tau}_{n}\to\infty.$\footnote{While $\tilde{\tau}_{n}$ depends on $x$, it does not create an issue.
For notational simplicity, we will be implicit about $\tilde{\tau}_{n}$'s
dependence on $x$ below.} Then, for any $x\in X$, 
\begin{align}
P\Big(\inf_{b\in\Pi(x)}\hat{\xi}_{n}(b,x)<-1\Big) & \le P\Big(\sup_{b\in\Pi(x)}\frac{1}{\sqrt{n}}\sum_{\ell=1}^{n}Z_{\ell}(b)-E[Z_{\ell}(b)]>\underline\zeta\tilde{\tau}_{n}\Big)\label{eq:maximal0}\\
 & =P\Big(\sup_{b\in\Pi(x)}\mathcal{Z}_{n}(b)>\underline\zeta\tilde{\tau}_{n}\Big)\\
 & =P\Big(\sup_{b\in\Pi(x)}\{\mathcal{Z}_{n}(b)-\mathcal{Z}_{n}(0)\}>\underline\zeta\tilde{\tau}_{n}\Big),\label{eq:maximal1}
\end{align}
where $\mathcal{Z}_{n}(0)=0$ by definition. By Lemma \ref{lem:subGaussian},
$\mathcal{Z}_{n}$ is a sub-Gaussian process. The chaining tail inequality
\citep[Theorem 5.29]{Handel:2018aa} then implies, for any $u\ge0$,
\begin{align}
P\Big(\sup_{b\in\Pi(x)}\{\mathcal{Z}_{n}(b)-\mathcal{Z}_{n}(0)\}>C\int_{0}^{\infty}\sqrt{\ln N(B,\|\cdot\|,\epsilon)}d\epsilon+u\Big)\le C\exp(-\frac{u^{2}}{CD^{2}}),\label{eq:maximal2}
\end{align}
where $C$ is a universal constant, $D=\text{diam}(B)$, and $N(B,\|\cdot\|,\epsilon)$
is the covering number of $B$. Note that 
\begin{align}
M\equiv\int_{0}^{\infty}\sqrt{\ln N(B,\|\cdot\|,\epsilon)}d\epsilon\le\int_{0}^{D}\sqrt{|Y|\ln(D/\epsilon)}d\epsilon<\infty,\label{eq:maximal3}
\end{align}
where we used $N(B,\|\cdot\|,\epsilon)\le(\text{diam}(B)/\epsilon)^{|Y|}$,
and $\ln N(B,\|\cdot\|,\epsilon)=0$ for all $\epsilon>\text{diam}(D)$.

Let $u_{n}=\underline\zeta\tilde{\tau}_{n}-M$ and observe that $u_{n}>0$
for all $n\ge \bar{N}$ sufficiently large. By \eqref{eq:maximal1}-\eqref{eq:maximal3},
for any $n\ge \bar{N},$ 
\begin{align}
P\Big(\sup_{b\in\Pi(x)}\{\mathcal{Z}_{n}(b)-\mathcal{Z}_{n}(0)\}>\underline\zeta\tilde{\tau}_{n}\Big)\le C\exp(-\frac{u_{n}^{2}}{CD^{2}})
\end{align}
By \eqref{eq:maximal0}-\eqref{eq:maximal1} and the union bound,
we obtain 
\begin{align}
P(A_{n})\le\sum_{x\in X}P\Big(\inf_{b\in\Pi(x)}\hat{\xi}_{n}(b,x)<-1\Big)\le\sum_{x\in X}C\exp(-\frac{u_{n}^{2}}{CD^{2}}).
\end{align}
The claim of the lemma then follows due to $u_{n}\to\infty$ and $C,D$
being independent of $P$. \end{proof}

\begin{lemma}\label{lem:subGaussian} For each $b\in \ball$, let $Z_{\ell}(b)=-I_{x}'b$
and let $\mathcal{Z}_{n}(b)=\frac{1}{\sqrt{n}}\sum_{\ell=1}^{n}Z_{\ell}(b)-E[Z_{\ell}(b)]$.
Then, there exists $K\ge0$ such that 
\begin{align}
\|\mathcal{Z}_{n}(b)-\mathcal{Z}_{n}(\tilde{b})\|_{\psi_{2}}\le K\|b-\tilde{b}\|,~\text{ for all }b,\tilde{b}\in \ball.
\end{align}
\end{lemma} \begin{proof} Note that 
\begin{align}
Z_{\ell}(b)-Z_{\ell}(\tilde{b}) & \le|Z_{\ell}(b)-Z_{\ell}(\tilde{b})|\le\|I_{x}\|\|b-\tilde{b}\|,\label{eq:subGaussian1}
\end{align}
and $W=\|I_{x}\|$ is a bounded random variable, which implies $\|W\|_{\psi_{2}}\le\frac{1}{\sqrt{\ln2}}\|W\|_{\infty}$
\citep[p.25]{Vershynin2018}. The inequality in \eqref{eq:subGaussian1}
also implies $E[Z_{\ell}(b)]-E[Z_{\ell}(\tilde{b})]\le E[|W|]\|b-\tilde{b}\|\le\|W\|_{\infty}\|b-\tilde{b}\|.$
Since $\|W\|_{\infty}\le\sqrt{|Y|}$, the triangle inequality implies
\begin{align}
\big\|(Z_{\ell}(b)-E[Z_{\ell}(b)])-(Z_{\ell}(\tilde{b})-E[Z_{\ell}(\tilde{b})])\big\|_{\psi_{2}}\le\sqrt{|Y|}(1+\frac{1}{\sqrt{\ln2}})\|b-\tilde{b}\|.
\end{align}
Note that $\mathcal{Z}_{n}(b)-\mathcal{Z}_{n}(\tilde{b})$ is the
sum of independent mean-zero sub-Gaussian random variables $(Z_{\ell}(b)-E[Z_{\ell}(b)])-(Z_{\ell}(\tilde{b})-E[Z_{\ell}(\tilde{b})])$
and hence by Proposition 2.6.1 in \citet{Vershynin2018}, 
\begin{align}
\|\mathcal{Z}_{n}(b)-\mathcal{Z}_{n}(\tilde{b})\|_{\psi_{2}}^{2} & =\frac{1}{n}\Big\|\sum_{\ell=1}^{n}(Z_{\ell}(b)-E[Z_{\ell}(b)])-(Z_{\ell}(\tilde{b})-E[Z_{\ell}(\tilde{b})])\Big\|_{\psi_{2}}^{2}\\
 & \le\frac{1}{n}C\sum_{\ell=1}^{n}\big\|(Z_{\ell}(b)-E[Z_{\ell}(b)])-(Z_{\ell}(\tilde{b})-E[Z_{\ell}(\tilde{b})])\big\|_{\psi_{2}}^{2}\\
 & \le C\big(\sqrt{|Y|}(1+\frac{1}{\sqrt{\ln2}})\|b-\tilde{b}\|\big)^{2},
\end{align}
where $C>0$ is a universal constant. The claim then follows with
$K=\sqrt{C}(1+\frac{1}{\sqrt{\ln2}}).$ \end{proof}

\begin{lemma}\label{lem:Hausdorff_consistency}
Suppose Assumptions \ref{as:primitives}-\ref{as:Wn}	hold. 
Then, for any $\epsilon>0$ and $\eta>0$, there exists $\bar N_{\epsilon,\eta}\in\mathbb N$ such that
\begin{align}
\sup_{P\in\mathcal P}P\big(\sup_{x\in X}d_H(\balln,\ball)>\epsilon\big)<\eta.	
\end{align}
for all $n\ge \bar N_{\epsilon,\eta}$.
\end{lemma}

\begin{proof}
We show convergence using the following isometry \citep[][Theorem 1.1.12]{Li:2002aa}:
\begin{align}
d_H(\balln,\ball)=	\sup_{p\in \mathbb S^{|Y|-1}}|h(p,\balln)-h(p,\ball)|.\label{eq:Hormander}
\end{align}
For this, recall that the support function of $\ball$ is given by
\begin{align}
h(p,\ball)=\sup_{b\in\ball}p^{\top}b=\sup&~ p^{\top}b.\notag\\
s.t.&~b^{\top}W_xb\le 1
\end{align}	
Solving the quadratic program above yields
\begin{align}
h(p,\ball)=(p^{\top}W_x^{-1}p)^{1/2},
\end{align}
where $W_x^{-1}$ is well defined by $\underline{\lambda}(W_x)>0$.
A similar argument can show
\begin{align}
h(p,\balln)=(p^{\top}W^{-1}_{n,x}p)^{1/2}.		
\end{align}
By the Cauchy-Schwarz inequality and $p$ being a unit vector, it follows that
\begin{align}
|h(p,\balln)^2-	h(p,\ball)^2|&=|p^{\top}(W^{-1}_{n,x}-W_x^{-1})p|\notag\\
&=|p^{\top}W_{n,x}^{-1}(W_x-W_{n,x})W_x^{-1}p|\notag\\
&\le \|p^{\top}W_{n,x}^{-1}\|\|(W_x-W_{n,x})W_x^{-1}p\|\notag\\
&\le \|W_{n,x}^{-1}\|_{op}\|(W_x-W_{n,x})\|_{op}\|W_x^{-1}p\|\notag\\
&\le \|W_{n,x}^{-1}\|_{op}\|W_x^{-1}\|_{op}\|W_x-W_{n,x}\|_{op}.\label{eq:Hausdorff1}
\end{align}

Note that $\|W_x^{-1}\|_{op}=\underline{\lambda}(W_x)^{-1}\le \underline{\kappa}^{-1}$. Let $\varepsilon= \underline{\kappa}^{-1}-\underline{\lambda}(W_x)^{-1}.$ By the Lipschitz continuity of $\underline{\lambda}$ for Hermitian matrices \citep[][Corollary III.2.6]{Bhatia_1997}, there is $\delta>0$ such that $$\|W_x-W_{n,x}\|_{op}\le \delta~\Rightarrow~\underline{\lambda}(W_{n,x})^{-1}\le \underline{\lambda}(W_x)^{-1}+\varepsilon=\underline{\kappa}^{-1}.$$ Let $\epsilon'=\min\{\delta,\frac{2\underline{\underline\kappa}^2\epsilon}{\overline\kappa}\}$. Then, there is $N_{\epsilon',\eta}$ such that $\|W_x-W_{n,x}\|_{op}\le\epsilon'$ so that
\begin{align}
	|h(p,\balln)^2-	h(p,\ball)^2|\le \underline{\kappa}^{-2}\|W_x-W_{n,x}\|_{op}\le \frac{2\epsilon}{\overline\kappa},\label{eq:Hausdorff2}
\end{align}
with probability at least $1-\eta$ uniformly across $P$ for all $n\ge N_{\epsilon',\eta}$, where the first inequality follows from \eqref{eq:Hausdorff1}. 

Note also that 
\begin{align}
h(p,\ball)^2=p^{\top}W_x^{-1}p\ge \underline \lambda(W_x^{-1})=\overline \lambda(W_x)^{-1}\ge \overline{\kappa}^{-1}.\label{eq:Hausdorff3}
\end{align}
Similarly, again by the Lipschitz continuity of $\overline \lambda$ for real symmetric matrices and letting $\varepsilon'=\lambda(W)^{-1}-\overline{\kappa}^{-1}$, there exists $\delta'>0$ such that
\begin{align}
	\|W_x-W_{n,x}\|_{op}\le \delta'~\Rightarrow~\overline \lambda(W_{n,x})^{-1}\ge \overline \lambda(W_x)^{-1}-\varepsilon'=  \overline{\kappa}^{-1}
\end{align}
implying 
\begin{align}
h(p,\balln)^2	=p^{\top}W_{n,x}^{-1}p\ge \underline \lambda(W_{n,x}^{-1})=\overline \lambda(W_{n,x})^{-1}\ge \overline{\kappa}^{-1}.\label{eq:Hausdorff4}
\end{align} 
Let $\epsilon''=\min\{\delta',\epsilon'\}$. Noting that $A^2-B^2=(A-B)(A+B)$ and by \eqref{eq:Hausdorff2}, \eqref{eq:Hausdorff3}, and \eqref{eq:Hausdorff4}, there is $N_{\epsilon'',\eta}$ such that
\begin{align}
	|h(p,\balln)-	h(p,\ball)|=\frac{|h(p,\balln)^2-	h(p,\ball)^2|}{h(p,\balln)+	h(p,\ball)}\le \frac{\overline{\kappa}}{2} \times \frac{2\epsilon}{\overline\kappa}=\epsilon,
\end{align}
with probability at least $1-\eta$ uniformly across $P$ for all $n\ge N_{\epsilon'',\eta}$. Note that the bound in the above expression does not depend on $p$ nor $x$, and hence it is uniform across $p\in\mathbb S^{|Y|-1}$ and $x\in X$. The conclusion of the lemma then follows from the isometry in \eqref{eq:Hormander} and letting $\bar N_{\epsilon,\eta}=N_{\epsilon'',\eta}$.
\end{proof}

Let $Z=Y\times X$. Let $I_x:Z\times \{0,1\}^{|Y|}$ be defined by
 $I_{x}(w)=(1\{y_{\ell}=y_{1},x_{\ell}=x\},\dots,1\{y_{\ell}=y_{|Y|},x_{i}=x\})'$.  Define
\begin{align}
	\mathcal M_{P,x}\equiv\Big\{f_{\tilde b,x}:f_{\tilde b,x}(z)=\frac{\tilde b^{\top}I_x(z)}{p_x},\tilde b\in\tilde B_x\Big\},
\end{align}
and let $\mathcal M_P=\bigcup_{x\in X}\mathcal M_{P,x}$, where note that $X$ is finite. The following lemma characterizes $\mathcal M_P$'s uniform entropy.

\begin{lemma}\label{lem:uniform_entropy}
Suppose Assumptions \ref{as:Wn} (ii) and \ref{as:px} hold.
Then, there exist constants $K,v>0$ that do not depend on $P$ such that
\begin{align}
	\sup_{Q}N(\epsilon\|F\|_{L^2_Q},\mathcal M_P,L^2_Q)\le K\epsilon^{-v},~0<\epsilon<1,
\end{align}
where the supremum is taken over all discrete distributions, and $F$ is the envelope function for $\mathcal M_P.$
\end{lemma}

\begin{proof}
We first construct an envelope function. Observe that
\begin{align}
|f_{b,x}(z)|\le \underline{\zeta}^{-1}|b^{\top}I_x(z)|\le  \underline{\zeta}^{-1}\|b\|\|I_x(z)\|\le  \underline{\zeta}^{-1}\sqrt{|Y|}\|b\|,	
\end{align}
where the second inequality is due to the Cauchy-Schwarz inequality. Note that
\begin{align}
\sup_{b\in\ball}\|b\|^2 =&\sup_{\bar b^{\top}\bar b\le 1} ~\bar b^{\top} W^{-1}\bar b\le \underline{\lambda}(W)^{-1}\le \kappa^{-1},
\end{align}
where $\bar b=W^{1/2}b.$ Therefore, one can take $F(z)=\kappa^{-1/2}\underline{\zeta}^{-1}\sqrt{|Y|}$ as the envelope for $\mathcal M_P$.

Let $x\in X$ be fixed. Then,
	\begin{align}
		|f_{b,x}(z)-f_{b^{\top},x}(z)|\le \underline\zeta^{-1}\|I_x(z)\|\|b-b^{\top}\|\le C\underline{\zeta}^{-1}\|I_x(z)\|\|b-b^{\top}\|_{W_x},
	\end{align}
for some $C>0$, where the second inequality follows from the equivalence of norms in a Euclidean space. Note also that $\|I_x(z)\|\le |Y|^{1/2}$. Following the argument in the proof of Theorem 2 in \cite{Andrews:1994aa}, it follows that
\begin{align}
	\sup_{Q}N(\epsilon\|F\|_{L^2_Q},\mathcal M_{P,x},L^2_Q)\le K\epsilon^{-v},0<\epsilon<1,
\end{align}
with $v=|Y|.$ Note that $\mathcal M_P$ is a finite union of $\mathcal M_{P,x},x\in X$. The conclusion of the lemma then follows by arguing as in the proof of Theorem 3 in \cite{Andrews:1994aa} (see Eq. (A.4)).
\end{proof}

We define the variance semimetric $\rho_P\big((b_1,x_1),(b_2,x_2)\big)$ pointwise by
\begin{align}
	\rho_P\big((b_1,x_1),(b_2,x_2)\big)=\text{Var}\big(f_{b_1,x_1}(z)-f_{b_2,x_2}(z)\big)^{1/2}.
\end{align}

\begin{lemma}\label{lem:variance_semimetric}
	For any $(b_1,x_1), (b_2,x_2)\in K\times X$, let $\rho\big((b_1,x_1),(b_2,x_2)\big)=\|b_1-b_2\|+d(x_1,x_2),$ where $d$ is the discrete metric. Let $0<\delta<1$. Suppose that $\rho\big((b_1,x_1),(b_2,x_2)\big)\le \delta$. Then, there exists $C>0$ such that,  for any $P\in\mathcal P$,
\begin{align}
	\rho_P\big((b_1,x_1),(b_2,x_2)\big)\le C\delta.
\end{align}
\end{lemma}

\begin{proof}
Let $0<\delta<1.$	Suppose that $\rho\big((b_1,x_1),(b_2,x_2)\big)\le \delta$. Since $d$ is the discrete metric, it must be the case that $\|b_1-b_2\|\le \delta$ and $x_1=x_2=\bar x$ for some $\bar x\in X$. By elementary calculation, 
\begin{align}
	\text{Var}\big(f_{b_1,x_1}(z)-f_{b_2,x_2}(z)\big)\le E[|f_{b_1,x_1}(z)-f_{b_2,x_2}(z)|^2]+E[f_{b_1,x_1}(z)-f_{b_2,x_2}(z)]^2\label{eq:var_semi1}
\end{align}
Observe that
\begin{align}
E[|f_{b_1,x_1}(z)-f_{b_2,x_2}(z)|^2]= E\big[\big((b_1-b_2)'I_{\bar x}(z)\big)^2\big]	\le E[\|b_1-b_2\|^2 \|I_{\bar x}(z)\|^2]\le |Y|^2\|b_1-b_2\|^2.
\end{align}
Note also that
\begin{align}
|E[f_{b_1,x_1}(z)-f_{b_2,x_2}(z)]|&=|	(b_1-b_2)'E[I_{\bar x}(z)]|\\
&\le \|b_1-b_2\|\|E[I_{\bar x}(z)]\|\le |Y|\|b_1-b_2\| ,
\end{align}
where the second inequality is by the Cauchy-Schwarz inequality. Hence, the second term in \eqref{eq:var_semi1} is bounded by $|Y|^2\|b_1-b_2\|^2$. These results imply 
\begin{align}
\text{Var}\big(f_{b_1,x_1}(z)-f_{b_2,x_2}(z)\big)^{1/2}\le \sqrt 2|Y|\|b_1-b_2\|\le C\delta,	
\end{align}
where $C=\sqrt 2|Y|.$ The conclusion of the lemma then follows.
\end{proof}

Below, for any $f\in \ell^\infty(K\times X)$, we write $Pf$ to denote its expectation with respect to $P$. Note that we may write the bootstrapped empirical process $\mathbb G^*_n$ as
\begin{align}
    \mathbb G^*_n(b,x)=\frac{1}{\sqrt n}\sum_{\ell=1}^n(M_{n\ell}-1)\frac{b^{\top}I_x(z_\ell)}{p_x},
\end{align}
where, for each $\ell$, $M_{n\ell}$ is the number of times that $z_\ell=(y_\ell,x_\ell)$ is redrawn from the original sample. We let $P_{M_n}(\cdot|\{y_\ell,x_\ell\}_{\ell=1}^\infty)$ be the conditional law of $M_n=(M_{n1},\dots, M_{nn})'$ conditional on the sample path $\{y_\ell,x_\ell\}_{\ell=1}^\infty$ (see \cite{van_der_Vaart_1996} Ch. 3.6).

\begin{lemma}\label{lem:boot_se}
Suppose Assumptions \ref{as:primitives}-\ref{as:sampling}, and \ref{as:px} hold.
Then, for any $\epsilon>0$ and $\eta>0$, there exists $\delta>0$ and $N_{\epsilon,\eta}$ such that
\begin{align}
P_{M_n}\Big( \sup_{\rho_P((b_1,x_1),(b_2,x_2))\le \delta}\big|\mathbb{G}_{n}^{*}(b_1,x_1)- \mathbb{G}_{n}^{*}(b_2,x_2)\Big|\ge \eta\big|\{y_\ell,x_\ell\}_{\ell=1}^\infty\Big)	\le \epsilon
\end{align}
for all $n\ge N_{\epsilon,\eta}$ and for almost all sample paths $\{y_\ell,x_\ell\}_{\ell=1}^\infty$.
\end{lemma}

\begin{proof}
Observe that
\begin{align}
\mathbb{G}_{n}^{*}(b,x)=\sqrt n\Big(\frac{1}{n}\sum_{\ell=1}^n	\frac{b^{\top}I_x(z^*_\ell)}{p_x}-E[\frac{b^{\top}I_x(z_\ell)}{p_x}]\Big).
\end{align}	
Below, we mimic the argument in \cite{van_der_Vaart_1996} (Ch.2.5) to show the  stochastic equicontinuity of empirical processes. For any $\delta>0$, define
\begin{align}
\mathcal M_{P,\delta}\equiv	\Big\{ f_{b_1,x_1}(z)-f_{b_2,x_2}(z)~\big|~\rho_P((b_1,x_1),(b_2,x_2))\le \tilde\delta,b_j\in \mathbb B_{x_j}, x_j\in X, j=1,2\Big\}.
\end{align}
Let $Z^*_n(\delta)=\sup_{f\in\mathcal M_{P,\delta}}|\sqrt n(\hat P^*_n-\hat P_n)f|.$
Note that by Lemma \ref{lem:variance_semimetric},
\begin{align}
\rho\big((b_1,x_1),(b_2,x_2)\big)\le \delta~\Rightarrow~\rho_P\big((b_1,x_1),(b_2,x_2)\big)\le \tilde\delta,
\end{align}
for $\tilde\delta=C\delta$ with a uniform constant $C>0$. It then follows that
\begin{align}
P_{M_n}\Big( \sup_{\rho((b_1,x_1),(b_2,x_2))\le \delta}\big|\mathbb{G}_{n}^{*}(b_1,x_1)- \mathbb{G}_{n}^{*}(b_2,x_2)\big|\ge \eta\big|\{y_\ell,x_\ell\}_{\ell=1}^\infty\Big)\le P_{M_n}\Big(Z^*_n(\tilde\delta)>\eta\big|\{y_\ell,x_\ell\}_{\ell=1}^\infty\Big).
\end{align}
By Markov's inequality and Lemma 2.3.1 in \cite{van_der_Vaart_1996}, one has
\begin{align}
P_{M_n}\big(Z^*_n(\tilde\delta_n)>\eta\big|\{y_\ell,x_\ell\}_{\ell=1}^\infty\big)&\le \frac{2}{\eta}E_{P_{M_n}\times P^e}\left[\sup_{f\in\mathcal M_{P,\tilde\delta}}\Big|\frac{1}{\sqrt n}\sum_{\ell=1}^ne_\ell f(z^b_\ell)\Big| \Bigg|\{y_\ell,x_\ell\}_{\ell=1}^\infty\right]\\
&=\frac{2}{\eta}E_{P_{M_n}}\left[E_{P^e}\left[\sup_{f\in\mathcal M_{P,\tilde\delta}}\Big| \frac{1}{\sqrt n}\sum_{\ell=1}^ne_\ell f(z^b_\ell)\Big| \Bigg|\{z^b_i=\ell\},\{y_\ell,x_\ell\}_{\ell=1}^\infty\right]\Bigg|\{y_\ell,x_\ell\}_{\ell=1}^\infty\right],\label{eq:sr10}
\end{align}
where $\{e_\ell\}_{\ell=1}^n$ are i.i.d. Rademacher random variables independent of $\{y_\ell,x_\ell\}_{\ell=1}^\infty$ and $\{M_{n}\}$. By Hoeffding's inequality, the stochastic process $f\mapsto\{n^{-1/2}\sum_{\ell=1}^ne_if(z_\ell)\}$ is sub-Gaussian for the $L^2_{\hat P_n}$ seminorm $\|f\|_{L^2_{\hat P_n}}=(n^{-1}\sum_{\ell=1}^n f(z_\ell)^2)^{1/2}.$ By the maximal inequality (Corollary 2.2.8) and arguing as in the proof of Theorem 2.5.2 in in \cite{van_der_Vaart_1996}, one then has
\begin{align}
	E_{P^e}\left[\sup_{f\in\mathcal M_{\tilde\delta}}\Bigl\vert\frac{1}{\sqrt n}\sum_{\ell=1}^ne_\ell f(z_\ell^b)\Bigr\vert \right]&\le K\int_0^{\tilde\delta}\sqrt{\ln N(\epsilon,\mathcal M_{P,\tilde\delta},L^2_{\hat P_n})}d\epsilon\notag\\
	&\le K\|F\|_{n}\int_0^{\tilde\delta/\|F\|_{n}}\sup_Q\sqrt{\ln N(\epsilon \|F\|_{L^2_Q},\mathcal M_P,L^2_Q)}d\epsilon\notag\\
	&\le K'\|F\|_{n}\int_0^{\tilde\delta/\|F\|_{n}}\sqrt{-v\ln\epsilon}d\epsilon,\label{eq:sr6}
\end{align}
for some $K'>0$, 
where $\|F\|_n=\sqrt{\frac{1}{n}\sum_{\ell=1}F^2(z^b_\ell)}$, and the last inequality follows from Lemma \ref{lem:uniform_entropy}. Note that $\sqrt{-\ln \epsilon}\ge -\ln \epsilon$ for $0<\epsilon <\tilde\delta/\|F\|_{n}$ with $\tilde\delta$ small enough and $\int_0^{\tilde\delta/\|F\|_{n}}-\ln \epsilon d\epsilon=\tilde\delta/\|F\|_{n}(1-\ln (\tilde\delta/\|F\|_{n})).$ Furthermore, by taking $F(w)=\kappa^{-1/2}\underline{\zeta}^{-1}\sqrt{|Y|}$ (see Proof of Lemma \ref{lem:uniform_entropy}), $\|F\|_{n}=\kappa^{-1/2}\underline{\zeta}^{-1}\sqrt{|Y|}=: \chi>0$ uniformly across $P$.  In sum, by choosing $\delta$ (and hence $\tilde\delta$) small enough, one has
\begin{align}
P_{M_n}\big(Z^*_n(\tilde\delta_n)>\eta\big|\{y_\ell,x_\ell\}_{\ell=1}^\infty\big)&\le\frac{2\tilde\delta}{\eta}(1-\ln (\tilde \delta/\chi))\le \epsilon,
\end{align}
for all $n$ sufficiently large.
This establishes the claim of the lemma.
\end{proof}

\begin{lemma}\label{lem:weak_convergence}
Suppose Assumptions \ref{as:primitives}-\ref{as:px}	hold. 
For each $n$, let $A_{1n}$ and $A_{2n}$ be defined as in \eqref{eq:def_A1n} and \eqref{eq:def_A2n} respectively. Then, for any continuity point of the distribution function of $\sup_{(b,x)\in\Psi_\infty}\mathbb{G}(b,x)$,
\begin{align}
P_{n}\bigg(\Big\{\big|P_{M_n}\big(\sup_{(b,x)\in\Psi_n}\mathbb{G}_{n}^{*}(b,x)\le a\big|y^n,x^n\big)-P\big(\sup_{(b,x)\in\Psi_\infty}\mathbb{G}(b,x)\le a\big)\big|>\epsilon\Big\} \cap A_{1n}^{c}\cap A_{2n}^c\bigg)\to 0.	
\end{align}
as $n\to\infty.$
\end{lemma} 
\begin{proof}
Define a metric $\rho$ on $K\times X$ by $\rho((b,x),(b^{\top},x'))=\|b-b^{\top}\|+d(x,x'),$ where $d$ is the discrete metric on $X$. We then let the Hausdorff distance on subsets of $K\times X$ be 
$$d_H(A,B)=\max\big\{\sup_{a\in A}\inf_{b\in B}\rho(a,b),\sup_{b\in B}\inf_{a\in A}\rho(a,b)\big\}.$$
For notational simplicity, we use $d_H$, the same notation as the Hausdorff distance for subsets of $K$ here.
Let 
\begin{align}
	D_n\equiv\{(b,x):b\in \balln,x\in X\},~~\text{ and }~~ D\equiv\{(b,x):b\in \ball,x\in X\}.
\end{align}
Suppose that, for any $x\in X$, $d_H(b_{n,x},b_x)<\delta$. Then, by the construction above and $d_H(X,X)=0$ due to $X$ being finite, we have $d_H(D_n,D)=\sup_{x\in X}d_H(b_{n,x},b)+d_H(X,X)<\delta$. Note that $\Psi_n$ and $\Psi_\infty$ can be expressed as subsets of $D_n$ and $D$ as follows
\begin{align}
	\Psi_n=D_n\cap\{(b,x)\in K \times X:\pi^*(b,x)=0\},~\text{ and }~\Psi_\infty=D\cap\{(b,x)\in K\times X:\pi^*(b,x)=0\}.
\end{align}
This therefore implies  $d_H(\Psi_n,\Psi_\infty)\le d_H(D_n,D)<\delta$. 

Let $(b^*_n,x^*_n)\in \text{argmax}_{(b,x)\in\Psi_n}\mathbb{G}_{n}^{*}(b,x).$ Let $\Pi_{\Psi_\infty}(b^*_n,x^*_n)$ be the projection of $(b^*_n,x^*_n)$ on $\Psi_\infty$ and note that $\|(b^*_n,x^*_n)-\Pi_{\Psi_\infty}(b^*_n,x^*_n)\|\le d_H(\Psi_n,\Psi_\infty)<\delta$. This implies
\begin{align}
\sup_{(b,x)\in\Psi_n}\mathbb{G}_{n}^{*}(b,x)-\sup_{(b,x)\in\Psi_\infty}\mathbb{G}_{n}^{*}(b,x)&\le \mathbb{G}_{n}^{*}(b^*_n,x^*_n)- \mathbb{G}_{n}^{*}(\Pi_{\Psi_\infty}(b^*_n,x^*_n))\notag\\
&\le \sup_{\rho((b,x),(b^{\top},x'))\le \delta}\big|\mathbb{G}_{n}^{*}(b,x)- \mathbb{G}_{n}^{*}(b^{\top},x')\big|
\end{align}
A similar argument gives
\begin{align}
\sup_{(b,x)\in\Psi_\infty}\mathbb{G}_{n}^{*}(b,x)- \sup_{(b,x)\in\Psi_n}\mathbb{G}_{n}^{*}(b,x)&\le	 \sup_{\rho((b,x),(b^{\top},x'))\le \delta}\big|\mathbb{G}_{n}^{*}(b,x)- \mathbb{G}_{n}^{*}(b^{\top},x')\big|.
\end{align}
Hence, for any $\eta>0$,
\begin{align}
\big|	\sup_{(b,x)\in\Psi_n}\mathbb{G}_{n}^{*}(b,x)-\sup_{(b,x)\in\Psi_\infty}\mathbb{G}_{n}^{*}(b,x)\big|\ge \eta~\Rightarrow  \sup_{\rho((b,x),(b^{\top},x'))\le \delta}\big|\mathbb{G}_{n}^{*}(b,x)- \mathbb{G}_{n}^{*}(b^{\top},x')\big|\ge \eta.\label{eq:se1}
\end{align}
Now suppose $(y^n,x^n)\in A_{1n}^{c}\cap A_{2n}^c$, where $A_{1n}$ and $A_{2n}$ are defined as in \eqref{eq:def_A1n} and \eqref{eq:def_A2n} respectively. Then, for any $\eta>0$, there is $\delta>0$ such that
\begin{align}
P_{M_n}&\big(\sup_{(b,x)\in\Psi_n}\mathbb{G}_{n}^{*}(b,x)\le a\big|y^n,x^n\big)\notag\\
&\le P\Big(\big|	\sup_{(b,x)\in\Psi_n}\mathbb{G}_{n}^{*}(b,x)-\sup_{(b,x)\in\Psi_\infty}\mathbb{G}_{n}^{*}(b,x)\big|\ge \eta\big|y^n,x^n\Big)
+ P_{M_n}\big(\sup_{(b,x)\in\Psi_\infty}\mathbb{G}_{n}^{*}(b,x)\le a+\eta\big|y^n,x^n\big)\notag\\
&\le P_{M_n}\big( \sup_{\rho((b,x),(b^{\top},x'))\le \delta}\big|\mathbb{G}_{n}^{*}(b,x)- \mathbb{G}_{n}^{*}(b^{\top},x')\big|\ge \eta\big|y^n,x^n\big)+P_{M_n}\big(\sup_{(b,x)\in\Psi_\infty}\mathbb{G}_{n}^{*}(b,x)\le a+\eta\big|y^n,x^n\big)\notag\\
&\le \frac{\epsilon}{3}+P_{M_n}\big(\sup_{(b,x)\in\Psi_\infty}\mathbb{G}_{n}^{*}(b,x)\le a+\eta\big|y^n,x^n\big),\label{eq:bound1}
\end{align}
for all $n$ sufficiently large, where the last inequality follows from \eqref{eq:se1} and Lemma \ref{lem:boot_se}. Therefore, by \eqref{eq:bound1} and the triangle inequality,
\begin{multline}
\Big|P_{M_n}\big(\sup_{(b,x)\in\Psi_n}\mathbb{G}_{n}^{*}(b,x)\le a\big|y^n,x^n\big)-P\big(\sup_{(b,x)\in\Psi_\infty}\mathbb{G}(b,x)\le a\big)\Big|\\
\le \frac{\epsilon}{3}+\Big|P_{M_n}\big(\sup_{(b,x)\in\Psi_\infty}\mathbb{G}_{n}^{*}(b,x)\le a+\eta\big|y^n,x^n\big)-P\big(\sup_{(b,x)\in\Psi_\infty}\mathbb{G}(b,x)\le a+\eta\big)\Big|\\
+ \Big|P\big(\sup_{(b,x)\in\Psi_\infty}\mathbb{G}(b,x)\le a+\eta\big)-P\big(\sup_{(b,x)\in\Psi_\infty}\mathbb{G}(b,x)\le a\big)\Big|\le \epsilon,\label{eq:bound2}
\end{multline}
for any $\eta>0$ and for all $n$ sufficiently large, where the last inequality follows from 
\begin{align}
	\Big|P_{M_n}\big(\sup_{(b,x)\in\Psi_\infty}\mathbb{G}_{n}^{*}(b,x)\le a+\eta\big|y^n,x^n\big)-P\big(\sup_{(b,x)\in\Psi_\infty}\mathbb{G}(b,x)\le a+\eta\big)\Big|\le \frac{\epsilon}{3},
\end{align}
for all $n$ sufficiently large by Theorem 3.6.2 in \cite{van_der_Vaart_1996} and the portmanteau theorem. Finally, since $a$ is a continuity point, one can choose $\eta$ sufficinetly small sothat
\begin{align}
	\Big|P\big(\sup_{(b,x)\in\Psi_\infty}\mathbb{G}(b,x)\le a+\eta\big)-P\big(\sup_{(b,x)\in\Psi_\infty}\mathbb{G}(b,x)\le a\big)\Big|\le \frac{\epsilon}{3}.
\end{align}
Hence,  \eqref{eq:bound2} establishes the claim of the lemma.
\end{proof}

\section{Details on the Monte Carlo Experiments}\label{sec:Appendix_MC}
\subsection{BNE Threshold}
Recall that $\sv(\tau)=P(\varepsilon_i\ge \tau)$.  
The threshold in the equilibrium strategy in \eqref{eq:BNE_threshold} solves
\begin{align}
\tau_i(\nu_i,t_i)
=
-\Bigl(
x^{\prime}\beta
+
\Delta
E_{\nu_{-i},t_{-i}}
\bigl[
\sv(\tau_{-i}(\nu_{-i},t_{-i}))
\mid
\nu_i,t_i
\bigr]
+
\nu_i
\Bigr),
\end{align}
where
\begin{multline}
E_{\nu_{-i},t_{-i}}
\bigl[
\sv(\tau_{-i}(\nu_{-i},t_{-i}))
\mid
t_i,\nu_i
\bigr]
\\=
\xi
E_{\nu_{-i}}
\bigl[
\sv(\tau_{-i}(\nu_{-i},\nu_i))
\mid
t_i
\bigr]
+
(1-\xi)
E_{\nu_{-i}}
\bigl[
\sv(\tau_{-i}(\nu_{-i},-\nu_i))
\mid
t_i
\bigr].
\end{multline}
Moreover,
\begin{align*}
E_{\nu_{-i}}
\bigl[
\sv(\tau_{-i}(\nu_{-i},\nu_i))
\mid
t_i=\eta
\bigr]
&=
\rho_{\eta}(\mu,\xi)\sv(\tau_{-i}(\eta,\nu_i))
+
\bigl(1-\rho_{\eta}(\mu,\xi)\bigr)\sv(\tau_{-i}(-\eta,\nu_i)),
\\
E_{\nu_{-i}}
\bigl[
\sv(\tau_{-i}(\nu_{-i},\nu_i))
\mid
t_i=-\eta
\bigr]
&=
\rho_{-\eta}(\mu,\xi)\sv(\tau_{-i}(-\eta,\nu_i))
+
\bigl(1-\rho_{-\eta}(\mu,\xi)\bigr)\sv(\tau_{-i}(\eta,\nu_i)).
\end{align*}
Hence, we obtain a system of equations that can be solved simultaneously for all thresholds $\tau_i(\nu_i,t_i)$. In particular,
\begin{align*}
\tau_i(\nu_i,\eta)
&=
-\Bigl(
x^{\prime}\beta
+
\Delta \xi
\bigl[
\rho_{\eta}(\mu,\xi)\sv(\tau_{-i}(\eta,\nu_i))
+
(1-\rho_{\eta}(\mu,\xi))\sv(\tau_{-i}(-\eta,\nu_i))
\bigr]
\\[-0.2em]
&\qquad
+
\Delta(1-\xi)
\bigl[
\rho_{\eta}(\mu,\xi)\sv(\tau_{-i}(\eta,-\nu_i))
+
(1-\rho_{\eta}(\mu,\xi))\sv(\tau_{-i}(-\eta,-\nu_i))
\bigr]
+
\nu_i
\Bigr),
\\[0.5em]
\tau_i(\nu_i,-\eta)
&=
-\Bigl(
x^{\prime}\beta
+
\Delta \xi
\bigl[
\rho_{-\eta}(\mu,\xi)\sv(\tau_{-i}(-\eta,\nu_i))
+
(1-\rho_{-\eta}(\mu,\xi))\sv(\tau_{-i}(\eta,\nu_i))
\bigr]
\\[-0.2em]
&\qquad
+
\Delta(1-\xi)
\bigl[
\rho_{-\eta}(\mu,\xi)\sv(\tau_{-i}(-\eta,-\nu_i))
+
(1-\rho_{-\eta}(\mu,\xi))\sv(\tau_{-i}(\eta,-\nu_i))
\bigr]
+
\nu_i
\Bigr).
\end{align*}

For a symmetric equilibrium, where thresholds are identical across firms, the system reduces to
\begin{align*}
\tau(\eta,\eta)
&=
-\Bigl(
x^{\prime}\beta
+
\Delta \xi
\bigl[
\rho_{\eta}(\mu,\xi)\sv(\tau(\eta,\eta))
+
(1-\rho_{\eta}(\mu,\xi))\sv(\tau(-\eta,\eta))
\bigr]
\\[-0.2em]
&\qquad
+
\Delta(1-\xi)
\bigl[
\rho_{\eta}(\mu,\xi)\sv(\tau(\eta,-\eta))
+
(1-\rho_{\eta}(\mu,\xi))\sv(\tau(-\eta,-\eta))
\bigr]
+
\eta
\Bigr),
\\[0.5em]
\tau(-\eta,\eta)
&=
-\Bigl(
x^{\prime}\beta
+
\Delta \xi
\bigl[
\rho_{\eta}(\mu,\xi)\sv(\tau(\eta,-\eta))
+
(1-\rho_{\eta}(\mu,\xi))\sv(\tau(-\eta,-\eta))
\bigr]
\\[-0.2em]
&\qquad
+
\Delta(1-\xi)
\bigl[
\rho_{\eta}(\mu,\xi)\sv(\tau(\eta,\eta))
+
(1-\rho_{\eta}(\mu,\xi))\sv(\tau(-\eta,\eta))
\bigr]
-
\eta
\Bigr),
\\[0.5em]
\tau(\eta,-\eta)
&=
-\Bigl(
x^{\prime}\beta
+
\Delta \xi
\bigl[
\rho_{-\eta}(\mu,\xi)\sv(\tau(-\eta,\eta))
+
(1-\rho_{-\eta}(\mu,\xi))\sv(\tau(\eta,\eta))
\bigr]
\\[-0.2em]
&\qquad
+
\Delta(1-\xi)
\bigl[
\rho_{-\eta}(\mu,\xi)\sv(\tau(-\eta,-\eta))
+
(1-\rho_{-\eta}(\mu,\xi))\sv(\tau(\eta,-\eta))
\bigr]
+
\eta
\Bigr),
\\[0.5em]
\tau(-\eta,-\eta)
&=
-\Bigl(
x^{\prime}\beta
+
\Delta \xi
\bigl[
\rho_{-\eta}(\mu,\xi)\sv(\tau(-\eta,-\eta))
+
(1-\rho_{-\eta}(\mu,\xi))\sv(\tau(\eta,-\eta))
\bigr]
\\[-0.2em]
&\qquad
+
\Delta(1-\xi)
\bigl[
\rho_{-\eta}(\mu,\xi)\sv(\tau(-\eta,\eta))
+
(1-\rho_{-\eta}(\mu,\xi))\sv(\tau(\eta,\eta))
\bigr]
-
\eta
\Bigr).
\end{align*}
\end{document}